\documentclass[journal]{IEEEtran}

\usepackage{amsmath}
\usepackage{amssymb}
\usepackage{amsfonts}
\usepackage{amsthm}
\usepackage{esint}
\usepackage[hidelinks]{hyperref}
\usepackage{empheq}
\usepackage{cite}
\usepackage{color}
\usepackage{dsfont} 
\usepackage[mathscr]{euscript} 

\ifCLASSOPTIONcompsoc
\usepackage[caption=false,font=normalsize,labelfont=sf,textfont=sf]{subfig}
\else
\usepackage[caption=false,font=footnotesize]{subfig}
\fi

\newtheorem{theo}{Theorem}

\newtheorem{lem}{Lemma}
\newtheorem{defi}{Definition}
\newtheorem{conj}{Conjecture}
\newtheorem{prop}{Proposition}
\newtheorem{rem}{Remark}

\newcommand{\pr}[3]{H(#1|#2)_{#3}}
\newcommand{\hr}[2]{H(#1)_{#2}}

\newcommand{\ket}[1]{\vert #1 \rangle} 
\newcommand{\kb}[2]{\left| #1 \vphantom{#2} \right>\left< #2 \vphantom{#1} \right|} 
\newcommand{\proj}[1]{\kb{#1}{#1}} 
\newcommand{\inner}[2]{\left\langle #1 \; , \; #2 \right\rangle} 

\newcommand{\bos}[1]{\boldsymbol{#1}}

\newcommand{\sx}{\mathcal{X}}
\newcommand{\sy}{\mathcal{Y}}

\newcommand{\tr}[1]{\mathrm{Tr} \left( #1 \right)}
\newcommand{\Tr}{\mathrm{Tr}}

\newcommand{\eps}{\varepsilon}
\newcommand{\dia}{\diamond}

\newcommand{\cb}{\mathrm{cb}}

\newcommand{\pptd}{\mathcal{PPT}_1}
\newcommand{\ppto}{\mathcal{PPT}_2}

\newcommand{\bhyx}{\mathcal{L}(\mathcal{H}_B \otimes \mathcal{H}_A)}

\newcommand{\bhx}{\mathcal{L}(\mathcal{H}_A)}
\newcommand{\dhx}{\mathcal{D}(\mathcal{H}_A)}
\newcommand{\bhy}{\mathcal{L}(\mathcal{H}_B)}

\newcommand{\ptx}[1]{{#1}^{\mathrm{T}_{A}}}

\newcommand{\trx}{\Tr_{A}}
\newcommand{\try}{\Tr_{B}}

\newcommand{\iyx}{\mathds{1}}
\newcommand{\ix}{\mathds{1}_{A}}
\newcommand{\iy}{\mathds{1}_{B}}

\begin{document}

\title{Tightening continuity bounds for entropies~\\ and bounds on quantum capacities}

\author{Michael~G~Jabbour
        and~Nilanjana~Datta
\thanks{Michael~G~Jabbour is with the Centre for Quantum Information and Communication, \'Ecole polytechnique de Bruxelles, CP 165/59, Universit\'e libre de Bruxelles, 1050 Brussels, Belgium. Nilanjana~Datta is with the Department of Applied Mathematics and Theoretical Physics, Centre for Mathematical Sciences, University of Cambridge, Cambridge CB3 0WA, United Kingdom.
}}

\maketitle

\begin{abstract} 
Uniform continuity bounds on entropies are generally expressed in terms of a single distance measure between probability distributions or quantum states, typically, the total variation- or trace distance. However, if an additional distance measure is known, the continuity bounds can be significantly strengthened. Here, we prove a tight uniform continuity bound for the Shannon entropy in terms of both the local- and total variation distances, sharpening an inequality in [I. Sason, IEEE Trans. Inf. Th., \textbf{59}, 7118 (2013)]. We also obtain a uniform continuity bound for the von Neumann entropy in terms of both the operator norm- and trace distances. We then apply our results to compute upper bounds on channel capacities. We first refine the concept of approximate degradable channels by introducing $(\varepsilon,\nu)$--degradable channels. These are $\varepsilon$--close in diamond norm and $\nu$--close in completely bounded spectral norm to their complementary channel when composed with a degrading channel. This leads to improved upper bounds to the quantum- and private classical capacities of such channels. Moreover, these bounds can be further improved by considering certain unstabilized versions of the above norms. We show that upper bounds on the latter can be efficiently expressed as semidefinite programs. As an application, we obtain a new upper bound on the quantum capacity of the qubit depolarizing channel.
\end{abstract}


\IEEEpeerreviewmaketitle

\section{Introduction}
\label{sec:intro}

Entropies play a fundamental role in Information Theory. They characterize optimal rates of various information-processing tasks. For example, the optimal rate of asymptotically reliable compression of information emitted by a discrete, memoryless classical information source (i.e., its data compression limit) is given by its {\em{Shannon entropy}}. Similarly, the {\em{von Neumann entropy}} of a memoryless quantum information source characterizes its data compression limit. There are a plethora of different entropies and they satisfy a number of interesting properties. One such property concerns the following question: {\em{Given two quantum states (or random variables) which are close to each other, with respect to some suitable distance measure, how close are their entropies?}} The closeness of the entropies are expressed in terms of so-called continuity bounds.
Continuity bounds for the Shannon- and von Neumann entropies were originally obtained with the distance measures being the total variation distance and the trace distance, respectively. These are stated below. Refer to Section~\ref{sec:prelims} for notations and definitions.

\begin{theo}\label{th-1}{\em[Csisz\'ar inequality~\cite{PetzCon}]} For two probability distributions $p_X$ and $q_X$ on a finite set ${\mathcal{X}}$,
\begin{equation}
\begin{aligned}
    & |H(p_X) - H(q_X)| \\
    & \quad \leq {\mathrm{TV}}(p_X , q_X ) \log(|{\mathcal{X}}|  - 1) + h({\mathrm{TV}}(p_X , q_X )),
\end{aligned}
\end{equation}
where $H(p_X)$ and $H(q_X)$ denote the Shannon entropies of $p_X$ and $q_X$, $TV(p_X , q_X)$ denotes the total variation distance between them, and $h(\cdot)$ denotes the binary entropy function. 
\end{theo}

\begin{theo}\label{th-2}{\em[Audenaert-Fannes-Petz inequality~\cite{Audenaert, Fannes, PetzCon}]} For two quantum states $\rho$ and $\sigma$ on a $d$-dimensional Hilbert space, 
\begin{equation} \label{eq:th-2}
    |S(\rho) - S(\sigma)| \leq T(\rho , \sigma ) \log(d  - 1) + h(T(\rho, \sigma)),
\end{equation}
where $S(\rho)$ and $S(\sigma)$ denote the von Neumann entropies of $\rho$ and $\sigma$, and $T(\rho, \sigma)$ denotes the trace distance between them.
\end{theo}

The continuity bounds in the above theorems are said to be {\em{uniform}} since they do not depend on the specifics of the probability distributions or the quantum states but instead depend only on the distances between them. Moreover, the bounds are {\em{tight}}, in the sense that there are pairs of probability distributions (resp.~states) for which the bounds are saturated\footnote{More precisely, the continuity bounds are tight in the sense that the supremum of the quotient of the left hand side and the right hand side of the inequality in Theorem~\ref{th-1} (resp.~Theorem~\ref{th-2}), taken over all pairs of probability distributions on $\mathcal{X}$ (resp.~states on $\mathcal{H}$) is equal to $1$.}.
In fact, the bound in Theorem~\ref{th-1} is achieved for the following pair of probability distributions: $p_X= \{1-t, \frac{t}{d-1}, \ldots, \frac{t}{d-1}\}$ and $q_X= \{1, 0, \ldots, 0\}$, where $t= TV(p_X,q_X)$ and $d = |{\mathcal{X}}|$. Correspondingly, the bound in Theorem~\ref{th-2} is achieved for a pair of (commuting) states whose eigenvalues are given by the above-mentioned probability distributions, with $t$ now being the trace distance between the states and $d$ being the dimension of the Hilbert space.

There are, however, other interesting and relevant distance measures between pairs of probability distributions or quantum states. Let $LO(p_X, q_X)$, defined through~\eqref{eq:LO}, denote the {\em{local distance}} between the probability distributions $p_X$ and $q_X$, and let $\|\rho - \sigma\|_\infty$, defined through~\eqref{eq:op}, denote the {\em{operator norm distance}} between the states $\rho$ and $\sigma$. It turns out that the optimizing probability distributions (resp.~states) for which equality holds in the above theorems satisfy the property that the local distance (resp.~the operator norm distance) is equal to the total variation distance (resp.~trace distance)~\cite{Sason2013}. This means that the continuity bounds can never be saturated for pairs of probability distributions or states for which these properties are not satisfied.

Realizing this fact, in the classical case, led Sason~\cite{Sason2013} to derive a refined continuity bound for the Shannon entropy in terms of both the  total variation distance and the local distance between the probability distributions. This bound is given by Lemma~\ref{lem:Sason} below. However, Sason's bound is tight only when the quotient of the total variation distance by the local distance is an integer.
\begin{lem}[Sason~\cite{Sason2013}] \label{lem:Sason}
    For two probability distributions, $p_X$ and $q_X$, on a finite set ${\mathcal{X}}$,
    \begin{equation} \label{eq:ContBoundShannon2}
    \begin{aligned}
        & \left| \hr{X}{p} - \hr{X}{q} \right| \\
        & \quad \leq \mathrm{TV}(p_X,q_X) \log(\beta |\sx|-1) + h(\mathrm{TV}(p_X,q_X)),
    \end{aligned}
    \end{equation}
    where
    \begin{equation}
        \beta \coloneqq \frac{\mathrm{LO}(p_X,q_X)}{\mathrm{TV}(p_X,q_X)}.
    \end{equation}
    Furthermore, this bound is tight when $1/\beta$ is an integer.
\end{lem}
\medskip

In this paper, we sharpen Sason's inequality (Lemma~\ref{lem:Sason}) by proving a tight, uniform continuity bound for the Shannon entropy in terms of both the total variation distance and the local distance between 
a pair of probability distributions (see Theorem~\ref{theo:main}). We then derive a quantum analogue of Lemma~\ref{lem:Sason}, namely,  a uniform continuity bound for the von Neumann entropy, in terms of both the trace distance and operator norm distance between a pair of quantum states (see Theorem~\ref{theo:ContinuityOpNorm}). The bound is valid under a constraint relating these two distances (stated in the theorem) and it is tight when the quotient of the trace distance by the operator norm distance is an integer. Next, we use this result to obtain upper bounds on the quantum and private classical capacity of channels, by generalizing the notion of approximate degradable channels introduced by Sutter et al.~\cite{Sutter2017}.
\medskip

\noindent
{\bf{Layout of the paper with a summary of our contributions:}} Relevant notations and definitions are given in Section~\ref{sec:prelims}. Our main results on continuity bounds on entropies are stated in Theorem~\ref{theo:main} and Theorem~\ref{theo:ContinuityOpNorm} of Section~\ref{sec:main-results}. 
The proofs of these theorems are given in Section~\ref{sec:proofs-entropies}. In Section~\ref{sec:quantum-cap} we introduce the notion of $(\eps, \nu)$-degradable channels and use the result stated in Theorem~\ref{theo:ContinuityOpNorm} to obtain upper bounds on the quantum and private classical capacity of such channels. In fact, we prove two sets of results: $(i)$ Theorems~\ref{theo:cohinfcont} and~\ref{theo:upperQuantCapa} are expressed in terms of the diamond norm- and cb-norm distances (between the complementary channel and the channel composed with a degrading channel). These norms are frequently used in quantum information theory; $(ii)$ Propositions~\ref{prop:cohinfcont2} and~\ref{prop:upperQuantCapa2}, on the other hand, give bounds on the quantum and private classical capacity in terms of efficiently computable\footnote{That is, in terms of semidefinite programs (SDPs).} upper bounds on the unstabilized versions of the diamond norm- and cb norm distances, with the suprema in their definitions being restricted to the set of density matrices. We include these Propositions because they are guaranteed to provide upper bounds on the quantum and private classical capacity which are at least as good as that obtained in~\cite{Sutter2017}. The SDP formulations of bounds on these latter norms are given in Proposition~\ref{prop:primalDualUpperBound1to1} and Proposition~\ref{prop:SDP-ub-infty} respectively, and might be of independent interest. As an application of our results, we consider the qubit depolarizing channel and show that Proposition~\ref{prop:upperQuantCapa2} yields a tighter bound on its quantum capacity in comparison to that given by Sutter et al.~\cite{Sutter2017} (see Figs.~\ref{fig:depol1} and~\ref{fig:depol2}).
We conclude the paper by stating a conjecture and some open problems in Section~\ref{sec:discussions}.

\section{Mathematical preliminaries}
\label{sec:prelims}

\textbf{Classical systems.} In this paper, we take logarithms to base $2$. Let $\sx \coloneqq \left\lbrace 1,2, \cdots, |\sx| \right\rbrace$ be a finite set, and let $\mathcal{P}_{\sx}$ denote the set of probability distributions on $\sx$.
The {\em{total variation distance}} between $p_X= \{p_X(x)\}_{x \in \sx}, q_X = \{q_X(x)\}_{x \in \sx}
\in \mathcal{P}_{\sx}$ is defined as
\begin{equation} \label{eq:TV}
	\mathrm{TV}(p_X,q_X) \coloneqq \frac{1}{2} \sum_{x \in \sx} |p_X(x) - q_X(x)|,
\end{equation}
while the {\em{local distance}} is defined as
\begin{equation} \label{eq:LO}
	\mathrm{LO}(p_X,q_X) \coloneqq \sup_{x \in \sx} |p_X(x) - q_X(x)|.
\end{equation}
The Shannon entropy of a discrete random variable $X$ with distribution $p_X$ is defined as~\cite{Shannon} $H(X)_p \coloneqq - \sum_{x \in \sx} p_X(x) \log p_X(x)$.
It belongs to a family of functions related to the concept of majorization. Given $\bos{u} \in \mathbb{R}^d$ for some finite dimension $d$, let $\bos{u}^{\downarrow} \in \mathbb{R}^d$ be the vector containing the elements of $\bos{u}$ arranged in non-increasing order. For $\bos{u}, \bos{v} \in \mathbb{R}^d$, we say $\bos{u}$ is majorized by $\bos{v}$, (denoted as $\bos{u} \prec \bos{v}$)~\cite{Majorization}, if
\begin{equation} \nonumber
    \sum_{j=1}^k \bos{u}^{\downarrow}_j \leq \sum_{j=1}^k \bos{v}^{\downarrow}_j, \quad \forall k = 1, \cdots, d-1, \quad \sum_{j=1}^d \bos{u}^{\downarrow}_j = \sum_{j=1}^d \bos{v}^{\downarrow}_j.
\end{equation}
A function $f:\mathbb{R}^d \to \mathbb{R}$ is said to be Schur convex if $f(\bos{u}) \leq f(\bos{v})$ for any pair $\bos{u}, \bos{v} \in \mathbb{R}^d$ with $\bos{u} \prec \bos{v}$, and it is said to be Schur concave if $(-f)$ is Schur convex. The Shannon entropy is a notable example of Schur concave functions when it is taken as a function of the vectors of probability distributions. In what follows, $\lceil \cdot \rceil$ represents the ceiling function, $\lfloor \cdot \rfloor$ is the floor function and, for $\eps \in [0,1]$, $h(\eps) \coloneqq -\varepsilon \log \varepsilon - (1-\varepsilon) \log (1-\varepsilon)$ denotes the binary entropy while $g(\eps) \coloneqq (1+\eps) \log(1+\eps) - \eps \log \eps$. The latter function is sometimes called \emph{bosonic entropy function}.
\medskip

\noindent
\textbf{Quantum systems.} Let $\mathcal{H}$ denote the Hilbert space associated to a finite-dimensional quantum system;  $\dim {\mathcal{H}} = d < \infty$.  Let ${\mathcal{L}}({\mathcal{H}})$ denote the algebra of linear operators acting on ${\mathcal{H}}$, and ${\mathcal{D}}(\mathcal{H}) \subset {\mathcal{L}}({\mathcal{H}})$  denote the set of states (density matrices) of the system, i.e.~positive semidefinite operators of unit trace.  Recall that the Schatten $p$-norm of any operator $A \in {\mathcal{L}}(\mathcal{H})$ for $p \in [1, \infty)$  is defined as ${\displaystyle \|A\|_{p}=[{\mathrm{Tr}}(|A|^{p})]^{\frac {1}{p}}},$ with $|A|= \sqrt{A^\dagger A}$ . Then ${\displaystyle \|A\|_{1} = {\mathrm{Tr}} |A| }$ is the trace norm, whereas the operator norm is given by
\begin{equation} \label{eq:op}
    \|A\|_\infty \coloneqq \lim_{p \to \infty}\|A\|_p= \sup \left\{\|A\ket{\psi}\| : \| \ket{\psi} \| =1, \ket{\psi} \in {\mathcal{H}} \right\}.
\end{equation}
The trace norm and operator norm are both multiplicative over tensor products~\cite{Watrousln}, i.e., $||A \otimes B||_i = ||A||_i ||B||_i$ for $i=1,\infty$ and $A,B \in {\mathcal{L}}(\mathcal{H})$. All Schatten $p$-norms, for $p \in [1, \infty]$, obey the duality relations~\cite{Watrousln}
\begin{equation} \label{eq:normduality}
    ||A||_p = \max \{ | \inner{B}{A} | : B \in {\mathcal{L}}(\mathcal{H}), ||B||_q \leq 1 \},
\end{equation}
where $1/p + 1/q = 1$ and $\inner{B}{A} \coloneqq \tr{B^{\dagger} A}$ is the Hilbert--Schmidt inner product on ${\mathcal{L}}(\mathcal{H})$.

The trace distance between two states $\rho, \sigma \in {\mathcal{D}}({\mathcal{H}})$ is given by $\mathrm{T}(\rho, \sigma) \coloneqq \frac{1}{2}||\rho-\sigma||_1$, while the operator norm distance between $\rho$ and $\sigma$ is given by $\|\rho-\sigma\|_\infty$.
By H\"older's inequality, ${\displaystyle \|A\|_1 \leq d \|A\|_\infty}$ for any $A \in {\mathcal{L}}({\mathcal{H}})$, which implies that
$\mathrm{T}(\rho, \sigma) \leq \frac{d}{2} \| \rho - \sigma \|_{\infty}.$
Further, by monotonicity of Schatten-$p$ norms,
$\|\rho - \sigma\|_{\infty} \leq \| \rho - \sigma \|_1 = 2 \mathrm{T}(\rho, \sigma)$.
However, both these inequalities can be strengthened to the inequalities~\eqref{st-op} of Lemma~\ref{lem:constDimQ} in Section~\ref{sec:main-results}.
The von Neumann entropy of a state $\rho$ is given by $S(\rho)\coloneqq - {\mathrm{Tr}}(\rho \log \rho)$.

A Hermiticity-preserving map, $\Phi$, from a system $A$ to a system $B$ (with associated Hilbert spaces $\mathcal{H}_A$ and $\mathcal{H}_B$) is a linear map: $\Phi : \mathcal{L}(\mathcal{H}_A) \rightarrow \mathcal{L}(\mathcal{H}_B)$ such that $\Phi(X) \in \mathcal{L}(\mathcal{H}_B)$ is Hermitian for all choices of $X \in \mathcal{L}(\mathcal{H}_A)$ Hermitian.
For notational simplicity, we denote such a map simply as $\Phi: A \rightarrow B$.
A quantum channel, $\Phi : A \rightarrow B$, is a linear, completely positive trace-preserving (CPTP) map.
A linear map $\tilde{\Phi} : A \rightarrow B$ is said to be completely positive and unital (CPU) if $\tilde{\Phi}({\mathds{1}}_A) = {\mathds{1}}_B$, where ${\mathds{1}}_A$ and ${\mathds{1}}_B$ denote the identity operators in ${\mathcal{L}}({\mathcal{H}}_A)$ and ${\mathcal{L}}({\mathcal{H}}_B)$, respectively.
The adjoint, $\Phi^*$, of a quantum channel $\Phi$ is defined through the relation $\tr{B \Phi(A)} = \tr{\Phi^*(B)A}$, and is a CPU map. By Stinespring's dilation theorem, one can choose an  isometry $V: \mathcal{H}_A \to \mathcal{H}_B \otimes \mathcal{H}_E$ such that $\Phi(\rho) = \Tr_E( V \rho V^\dagger)$ for any $\rho \in {\mathcal{D}}(\mathcal{H}_A)$.  We denote such an isometry simply as $V: A \to BE$. The complementary channel $\Phi^c$ is defined through the relation: $\Phi^c(\rho) = \Tr_B (V \rho V^\dagger)$ for any $\rho \in {\mathcal{D}}(\mathcal{H}_A)$. In this paper, the isometry $V$ is chosen such that the complementary channel is \emph{minimal}, that is, its output dimension is the least amongst all possible ones. The Choi--Jamio\l{}kowski representation of the linear map $\Phi: A \rightarrow B$ is the operator $J(\Phi) \in \mathcal{L}(\mathcal{H}_B \otimes \mathcal{H}_A)$ defined as
\begin{equation}
    J(\Phi) \coloneqq d \, (\Phi \otimes \rm{id}_A)(\proj{\Omega}),
\end{equation}
where $\ket{\Omega} \coloneqq \frac{1}{\sqrt{d}} \sum_{j=1}^d \ket{jj} \in {\mathcal{H}_A \otimes \mathcal{H}_A}$ denotes a maximally entangled state, and $d=\dim \mathcal{H}_A$. The operator $J(\Phi)$ is called the {\em{Choi matrix}} (or {\em{Choi operator)}} of the linear map. A linear map $\Phi: A \to B$ is completely positive (CP) if and only if $J(\Phi) \geq 0$, and it is trace-preserving (TP) if and only if $\Tr_B(J(\Phi)) = {\mathds{1}}_A$. For any input state $\rho\in \mathcal{D}(\mathcal{H}_A)$, the output of the linear map $\Phi: A \to B$ is given in terms of its Choi matrix as follows:
\begin{equation} \label{eq:Choi2}
    \Phi(\rho) = \Tr_A\left(J(\Phi) ({\rm{id}}_B \otimes \rho^T)\right),
\end{equation}
where the superscript $T$ denotes transposition, and the transpose is taken with respect to the basis chosen for the maximally entangled state $\ket{\Omega}$.

For a linear map $\Phi: A \to B$, one defines the superoperator norms (induced by Schatten norms) as follows: for $1 \leq p,q \leq \infty$,
\begin{equation} \nonumber
    \| \Phi \|_{q \rightarrow p} \coloneqq \max \left\{||\Phi(X)||_p : X \in {\mathcal{L}}({\mathcal{H}}_A), ||X||_q \leq 1 \right\}.
\end{equation}
However, more generally, one considers the map to be applied to only a subsystem of a larger system. In this case the following {\em{stabilized versions}} of the above norms are more relevant (see e.g.~\cite{Watrous2009}):
\begin{equation}
    ||| \Phi |||_{q \rightarrow p} \coloneqq \sup_{k \geq 1} \|\Phi \otimes {\rm{id}}_k \|_{q \rightarrow p}.
\end{equation}
In the above, ${\rm{id}}_k$ denotes the identity map: ${\rm{id}}_k :{\mathcal{L}}({\mathbb{C}}^k) \to {\mathcal{L}}({\mathbb{C}}^k)$.
Two important cases are the following:
$(i)$ the diamond norm (or completely bounded trace norm)
$\| \Phi\|_\diamond \coloneqq  ||| \Phi |||_{1 \rightarrow 1}$;
$(ii)$ the completely bounded norm (or completely bounded spectral (or operator) norm, cb-norm for short)
$\| \Phi \|_{\cb} \coloneqq  ||| \Phi |||_{\infty \rightarrow \infty}$.
It can be shown that~\cite{Watrous2009}
\begin{equation}
\|\Phi\|_\diamond = \| \Phi \otimes {\rm{id}}_{A}\|_{1 \rightarrow 1}
\,\, ; \,\,
\|\Phi\|_{\cb} = \| \Phi \otimes {\rm{id}}_{B}\|_{\infty \rightarrow \infty},
\end{equation}
where ${\rm{id}}_A$ (resp.~${\rm{id}}_B$) denote the identity maps on the systems $A$ and $B$ respectively; further,
\begin{equation} \label{eq:dualityDiamondCB}
\|\Phi \|_\diamond = \|\Phi^*\|_{\cb}.
\end{equation}
In the context of this work, we will also be interested in unstabilized versions of these norms, with the input space of the map being that of density matrices:
\begin{equation} \label{eq:normUnstabilized}
    \| \Phi \|_{1 \rightarrow p}^{\mathrm{D}} \coloneqq \max \{ ||\Phi(\rho)||_{p} : \rho \in \dhx \}.
\end{equation}
In what follows, we write $\| \Phi \|_p^{\mathrm{D}} \coloneqq \| \Phi \|_{1 \rightarrow p}^{\mathrm{D}}$ for notational simplicity. Similar unstabilized norms (where the input space of the map is that of Hermitian matrices) have been extensively studied in a series of works~\cite{Amosov2000,King2004,Watrous2004,Audenaert2005}. Norms such as the one defined in~\eqref{eq:normUnstabilized} are of particular interest in the context of quantum information processing, since the input to quantum channels are necessarily quantum states (i.e.~density matrices).

\section{Main results on continuity bounds for entropies}
\label{sec:main-results}

Our work was inspired by the paper of Sason~\cite{Sason2013}, in which the author derived a uniform continuity bound (see Lemma~\ref{lem:Sason} in the Introduction) on the Shannon entropy which depends on both the total variation distance- and the local distance between a pair of probability distributions. As mentioned in the Introduction, he derived this bound after realizing that the original, celebrated continuity bound given by Theorem~\ref{th-1} (and attributed to Cziszar) cannot be saturated by a pair of probability distributions $(p_X, q_X)$ when $\mathrm{LO}(p_X,q_X) \neq \mathrm{TV}(p_X,q_X)$.
Sason's bound, given by~\eqref{eq:ContBoundShannon2} in Lemma~\ref{lem:Sason}, is tight when $1/\beta$ is an integer. To see this, observe that the following pair of probability distributions $(\tilde{q}_{X}, \tilde{p}_{X})$ saturates the bound:
\begin{equation} \label{eq:probaTightSason1}
	\tilde{q}_{X}(x) \coloneqq \begin{cases}
		\nu/\eps, & \forall x \in \{ 1, \cdots, 1/\beta \}, \\
		0, & \textrm{else},
	\end{cases}
\end{equation}
and
\begin{equation} \label{eq:probaTightSason2}
	\tilde{p}_{X}(x) \coloneqq \begin{cases}
		(1-\eps)\nu/\eps, & \forall x \in \{ 1, \cdots, 1/\beta \}, \\
		\eps/(|\sx|-1/\beta), & \textrm{else},
	\end{cases}
\end{equation}
where $\eps \coloneqq \mathrm{TV}(\tilde{p}_X,\tilde{q}_X)$ and $\nu \coloneqq \mathrm{LO}(\tilde{p}_X,\tilde{q}_X)$. A simple calculation shows that $\left| \hr{X}{\tilde{p}} - \hr{X}{\tilde{q}} \right| = \hr{X}{\tilde{p}} - \hr{X}{\tilde{q}}$ attains the bound.
However, Sason's bound~\eqref{eq:ContBoundShannon2} is not tight when $1/\beta$ is not an integer. In order to show this, we prove a tight bound, and show that it is strictly lower than the one in~\eqref{eq:ContBoundShannon2} when $1/\beta$ is not an integer.

Define the set
\begin{equation} \label{eq:setnueps}
    \mathcal{S} \coloneqq \{ (\eps,\nu) \in [0, 1] \times [0, 1] : \nu \leq \eps \},
\end{equation}
and, for any integer $d > 1$, the function $f_d : \mathcal{S} \rightarrow [0, \infty)$ as
\begin{equation} \label{eq:RHSBound}
    f_d(\eps,\nu) \coloneqq h(\eps) + d_+ \, \nu \log \nu + \mu \log \mu + \eps \log d_- -\eps \log \eps,
\end{equation}
where $d_+ \in \mathbb{Z}$, $\mu \in [0,\nu)$ and $d_- \in \mathbb{Z}$ are defined through the following relations:
\begin{equation} \label{eq:rem}
    d_+ \coloneqq \lfloor \eps/\nu \rfloor, \quad \mu \coloneqq \eps - d_+ \nu, \quad d_- \coloneqq d - \lceil \eps/\nu \rceil.
\end{equation}
%
\begin{theo} \label{theo:main}
    Let $p_X, q_X \in \mathcal{P}_{\sx}$ with $d= |\sx|$, $t\coloneqq\mathrm{TV}(p_X,q_X) $ and $r\coloneqq\mathrm{LO}(p_X,q_X)$. Then the following inequality holds:
    \begin{equation} \label{eq:theo:ContinuityLocal}
        \left| \hr{X}{p} - \hr{X}{q} \right| \leq f_d(t,r),
    \end{equation}
    where the function $f_d$ is defined through~\eqref{eq:RHSBound}.
    Moreover, the inequality is tight.
\end{theo}
\noindent The proof of Theorem~\ref{theo:main} is given in Section~\ref{sec:proof:theo:main}.

\begin{rem}
    Observe that, for the distributions $p_X$ and $q_X$ satisfying the conditions of Theorem~\ref{theo:main},
    \begin{equation}
	   \eps = \frac{1}{2} \sum_{x \in \sx} |p_{X}(x)-q_{X}(x)| \leq \frac{1}{2} \sum_{x \in \sx} \nu = \frac{1}{2} |\sx| \nu,
    \end{equation}
    We also have that $\nu \leq \eps$ (see~\cite{Sason2013} for a simple proof), so that $\nu \in [2 \eps/|\sx|, \eps]$.
    Now, since $\nu |\sx| \geq 2 \eps$, 
    {Lemma~\ref{lem:MeLeqSason} in Appendix~\ref{sec:appA}}
    shows that the right hand side of~\eqref{eq:theo:ContinuityLocal} is strictly smaller 
    than that of~\eqref{eq:ContBoundShannon2} when $1/\beta = \eps/\nu$ is not an integer. Hence, the continuity bound given in Theorem~\ref{theo:main} is tighter than the one given in Lemma~\ref{lem:Sason}.
\end{rem}

In the quantum case, we obtain the following sharpening of the Audenaert-Fannes-Petz inequality stated in Theorem~\ref{th-2}.
\begin{theo} \label{theo:ContinuityOpNorm}
    Let $\rho, \sigma \in {\mathcal{D}}({\mathcal{H}})$, with $ {\mathrm{dim}}\, {\mathcal{H}} = d < \infty$.
    If $\mathrm{T}(\rho, \sigma) \leq \eps \in [0,1]$ and $||\rho-\sigma||_{\infty} \leq \nu \in [0,1]$, such that $\eps \leq \nu d/(\nu d + 3)$, then for $\beta \coloneqq \nu/\eps$,
    \begin{equation} \label{eq:theo:ContinuityOpNorm2}
	\left| S(\rho) - S(\sigma) \right| \leq \eps \log(\beta d-1) + h(\eps).
    \end{equation}
    Moreover, the inequality is tight when $1/\beta$ in an integer.
\end{theo}
\noindent Theorem~\ref{theo:ContinuityOpNorm} is proven in Section~\ref{sec:proof:theo:ContinuityOpNorm}.

\section{Proofs of continuity bounds on entropies}
\label{sec:proofs-entropies}

We employ the following lemma in the proofs of our results.

\begin{lem} \label{lem:constDimQ}
Consider a Hilbert space $\mathcal{H}$ with ${\rm{dim}}\, \mathcal{H}= d<\infty$. Let $\rho, \sigma \in {\mathcal{D}}({\mathcal{H}})$ such that $\eps \coloneqq \mathrm{T}(\rho, \sigma)$ and $\nu \coloneqq ||\rho-\sigma||_{\infty}$. Then we necessarily have that
\begin{equation}\label{st-op}
d \geq 2 \left\lceil \frac{\eps}{\nu} \right\rceil \quad \text{and} \quad \nu \leq \eps.
\end{equation}
\end{lem}
\noindent Note that a consequence of Lemma~\ref{lem:constDimQ} is that for ${\rm{dim}} \, {\mathcal{H}}=3$, we necessarily have $\nu=\eps$.
\begin{proof}[Proof of Lemma~\ref{lem:constDimQ}]
The difference $\rho-\sigma$ can be expressed as $\rho-\sigma = Q-S$, where $Q$ and $S$ are positive operators with mutually orthogonal supports, and $\tr{Q} = \tr{S} = \eps$. If we write the spectral decompositions
\begin{equation}
    Q = \sum_{i=1}^d \delta^+_i \proj{\psi^+_i}, \quad S = \sum_{i=1}^d \delta^-_i \proj{\psi^-_i},
\end{equation}
where $\delta^+_i, \delta^-_i \geq 0$ for all $i=1, \cdots, d$,
we have, by definition, that $\delta^+_i \leq \nu$ for all $i$ and $\delta^-_i \leq \nu$ for all $i$. We then have $\eps = \sum_{i} \delta^+_i \leq t_+ \nu$, where $t_+$ is the dimension of the support of $Q$.
Since $t_+$ is an integer, it implies $t_+ \geq \left\lceil \eps/\nu \right\rceil$. Similarly, $\eps = \sum_{i} \delta^-_i \leq t_- \nu$, where $t_-$ is the dimension of the support of $S$,
and since $t_-$ is an integer, $t_- \geq \left\lceil \eps/\nu \right\rceil$. From this, we get $d \geq t_+ + t_- \geq 2 \left\lceil \frac{\eps}{\nu} \right\rceil$, which gives the first statement.
To prove the second statement, note that since $\delta^+_i \geq 0$ for all $i$, $\eps \geq \delta^+_i$ for all $i$, so that $\eps \geq \sup_i \delta^+_i$. Similarly, $\delta^-_i \geq 0$ for all $i$, which implies $\eps \geq \delta^-_i$ for all $i$, and $\eps \geq \sup_i \delta^-_i$. From these last two facts, we get $\eps \geq \max \{ \sup_i \delta^+_i, \sup_i \delta^-_i \}= \nu$, which ends the proof.
\end{proof}

\subsection{Proof of Theorem~\ref{theo:main} \label{sec:proof:theo:main}}

Consider two probability distributions $p_{X}$ and $q_{X}$ in $\mathcal{P}_\sx$. In the context of this proof, we set $\eps \coloneqq \mathrm{TV}(p_X,q_X)$ and $\nu \coloneqq \mathrm{LO}(p_X,q_X)$ for notational simplicity. We assume, without loss of generality, that $\hr{X}{p} \geq \hr{X}{q}$ and define $\Delta H \coloneqq \hr{X}{p} - \hr{X}{q}$. The proof of Theorem~\ref{theo:main} consists of a series of steps which are described in detail below. The key idea is to alter the distributions $q_{X}$ and $p_{X}$, in a series of iterative steps such that the difference $\hr{X}{p}-\hr{X}{q}$ never decreases, and the total variation distance and local distance also remain unchanged. As a result of these manipulations, the entropy difference $\Delta H$ attains its maximal value under the constraints on the total variation distance and the local distance, yielding the desired upper bound.
\bigskip

\noindent \textbf{Intuitive idea behind each of the steps in the proof:} In Step~A, we reorder the elements of $q_{X}$ and $p_{X}$ in the following way: we group the elements for which $q_{X}(x) \geq p_{X}(x)$ together, and those for which $q_{X}(x) < p_{X}(x)$ together; in each group, we then reorder the elements of $q_X$ in non-increasing order, and reorder the elements of $p_X$ correspondingly. This preserves both the total variation distance and the local distance between them, as well as their Shannon entropies. In Step~B, we define a new probability distribution $q'_{X}$ that is obtained from $q_X$ by moving probability weights in it, in a manner which  decreases its Shannon entropy, while preserving the total variation distance to the probability distribution $p_{X}$. In Step~C, we define a joint distribution $q'_{XY}$ whose $\mathcal{X}$-marginal is $q'_{X}$, allowing us to lower bound the entropy, $H(X)_{q'}$, of $q'_{X}$ in terms of the conditional entropy $H(X|Y)_{q'}$ from which the contribution of the saturating distribution $\tilde{q}_{X}$ (defined through~\eqref{eq:probaTightMe1}) can be readily extracted. The probability distribution $p_{X}$ remains unchanged, as in the previous step. In Step~D, we derive an upper bound on the entropy, $H(X)_p$ of $p_{X}$. Finally, in Step~E, we combine the results of the previous steps to upper bound the entropy difference $\Delta H$.

Henceforth, we choose to represent a probability distribution $p_{X} \in \mathcal{P}_{\sx}$ as a column vector $\left\lbrace p_{X}(x) \right\rbrace_{x \in \sx}$ of dimension $|\sx|$, and hence use the notation $\left\lbrace p_{X}(x) \right\rbrace_{x \in \sx}\coloneqq  \left( p_{X}(1), \ldots, p_{X}(|\sx|) \right)$. Moreover, we use the symbols $q_{X}$ and $p_{X}$ to refer to both the probability distributions as well as their vector representations.
\medskip

\noindent \textbf{Step A:} Define the two sets
\begin{equation} \nonumber
        I \coloneqq \left\lbrace x : q_{X}(x) \geq p_{X}(x) \right\rbrace, \quad I^c \coloneqq \left\lbrace x : q_{X}(x) < p_{X}(x) \right\rbrace.
\end{equation}
Note that both $I$ and $I^c$ are non-empty, except in the trivial case where $q_X=p_X$ (for which Theorem~\ref{theo:main} is automatically satisfied). In both $p_{X}$ and $q_{X}$, arrange the elements corresponding to indices in $I$ ahead of those corresponding to indices in $I^c$. Then, in both $p_{X}$ and $q_{X}$, order the elements corresponding to indices in $I$ such that $q_{X}(x \! + \! 1) \leq q_{X}(x), \forall x \in I \setminus \left\lbrace |I| \right\rbrace$, and do the same for the elements corresponding to indices in $I^c$, so that we also have $q_{X}(x \! + \! 1) \leq q_{X}(x), \forall x \in I^c \setminus \left\lbrace |I^c| \right\rbrace$.
These operations correspond to permutations of elements within $p_{X}$ and $q_{X}$. As a consequence, neither $\hr{X}{p}$ nor $\hr{X}{q}$ changes under these operations. Moreover, both the total variation distance and the local distance between $p_X$ and $q_X$ remain unchanged, since the exact same permutations have been performed in $p_{X}$ and $q_{X}$.
\medskip

\noindent \textbf{Step B:} As stated in the definition of $f_d$ in~\eqref{eq:RHSBound},
define $d_+ \in \mathbb{Z}$ and $\mu \in [0,\nu)$ as $d_+ \coloneqq \lfloor \eps/\nu \rfloor$ and $\mu \coloneqq \eps - d_+ \nu$ (see~\eqref{eq:rem}).
This simply means that $\mu$ is the remainder of the division of $\eps$ by $\nu$, while $d_+$ is its quotient, i.e. $\eps = d_+ \nu + \mu$. Note here that if $\mu \neq 0$, then $d_+ + 1 \leq |I|$, or equivalently, $d_+ < |I|$. This can be seen with a proof by contradiction. Indeed, suppose $d_+ \geq |I|$, then $d_+ \nu \geq |I| \nu = \sum_{i=1}^{|I|} \nu$. Since $q_X(x)-p_X(x) \leq \nu$ for all $x \in I$, it implies that $d_+ \nu \geq \sum_{x=1}^{|I|} (q_X(x)-p_X(x)) = \eps$. But since $\mu \neq 0$, we know from~\eqref{eq:rem} that $d_+ \nu < \eps$, leading to a contradiction. Similarly, it can be shown by a proof by contradiction that if $\mu = 0$, then $d_+ \leq |I|$. As a consequence, we now look at two different cases: first, the case when $d_+<|I|$, for which $\mu \in [0,\nu)$, and secondly, the case when $d_+ = |I|$, for which we necessarily have $\mu=0$.

$(i)$ $d_+<|I|$: In that case, define the vector $q'_X$ whose elements satisfy the following relations:
\begin{equation} \label{eq:Bi}
	q'_{X}(x) \coloneqq \begin{cases}
		p_{X}(x) + \nu, & \forall x \in \{ 1, \cdots, d_+ \}, \\
		p_{X}(x) + \mu, & \textrm{if } x = d_+ + 1, \\
		p_{X}(x), & \forall x \in I \setminus \{ 1, \cdots, d_+ + 1 \}, \\
		q_{X}(x), & \forall x \in I^c.
	\end{cases}
\end{equation}
Note that $q'_{X}$ corresponds to a valid probability distribution, $\mathrm{TV}(p_X,q'_X) = \eps$ and $\mathrm{LO}(p_X,q'_X) = \nu$. Now, the elements of the vector $\{ q_X(x) \}_{x \in I}$ are arranged in non-increasing order. Furthermore, the elements of the vector $\{ q'_X(x) \}_{x \in I}$ can be written as
\begin{equation}
	q'_{X}(x) = \begin{cases}
		q_{X}(x) + \omega_x, & \forall x \in \{ 1, \cdots, d_+ \}, \\
		q_{X}(x) + \omega_x, & \textrm{if } x = d_+ + 1, \\
		q_{X}(x) - \omega_x, & \forall x \in I \setminus \{ 1, \cdots, d_+ + 1 \}.
	\end{cases}
\end{equation}
For $x \in \{ 1, \cdots, d_+ \}$, we have
\begin{equation}
	\omega_x = q'_{X}(x) - q_{X}(x) = \nu - (q_{X}(x) - p_{X}(x)) \geq 0,
\end{equation}
and for $x \in I \setminus \{ 1, \cdots, d_+ + 1 \}$, we have
\begin{equation}
	\omega_x = q_{X}(x) - q'_{X}(x) = q_{X}(x) - p_{X}(x) \geq 0.
\end{equation}
Finally, $\sum_{x=1}^{d_+ + 1} \omega_x - \sum_{x=d_+ + 2}^{|I|} \omega_x = 0$.
As a consequence, whatever the sign of $\omega_{d_+ + 1}$, we can employ 
{Lemma~\ref{lem:majq} in Appendix~\ref{sec:appA}}
in order to conclude that the vector $\{ q_X(x) \}_{x \in I}$ is majorized by the vector $\{ q'_X(x) \}_{x \in I}$. Using Lemma~5 of~\cite{JabbourDatta}, we then have that $q_{X} \prec q'_{X}$, which in turn implies that $\hr{X}{q'} \leq \hr{X}{q}$. We end up with
\begin{equation} \label{eq:StepBDeltaHi}
	\Delta H \leq \hr{X}{p} - \hr{X}{q'}.
\end{equation}

$(ii)$ $d_+ = |I|$: In that case, $\eps = |I| \nu$, so that
\begin{equation}
	\sum_{x=1}^{|I|} (q_X(x)-p_X(x)) = |I| \nu = \sum_{x=1}^{|I|} \nu,
\end{equation}
or,
\begin{equation}
	\sum_{x=1}^{|I|} (q_X(x)-p_X(x)-\nu) = 0.
\end{equation}
Since all the terms of the summation in the above equation are non-positive (which follows from the definition of $\nu$), we have that $q_X(x)=p_X(x)+\nu$ for all $x \in I$, which means that the elements of $q_X$ satisfy the following relations:
\begin{equation}
	q_{X}(x) = \begin{cases}
		p_{X}(x) + \nu, & \forall x \in \{ 1, \cdots, d_+ \}, \\
		q_{X}(x), & \forall x \in \{ d_++1, \cdots, |\sx| \}.
	\end{cases}
\end{equation}
In that case, we just define $q_X' \coloneqq q_X$, so that
\begin{equation} \label{eq:StepBDeltaHii}
	\Delta H = \hr{X}{p} - \hr{X}{q'}.
\end{equation}

\noindent \textbf{Step C:}
As mentioned before, our aim is to alter the distribution $q_X$ (in a manner such that the total variation distance and the local distance between it and $p_X$ are preserved) until it reaches the saturating distribution $\tilde{q}_X$ given by~\eqref{eq:probaTightMe1}. In the previous step, $q_X$ has was altered to $q'_X$. Next, we would like to ``extract" $\tilde{q}_X$ from $q'_X$. To do so, we define a joint distribution $q'_{XY}$ (whose ${\mathcal{X}}$-marginal is given by $q'_X$) such that $q'_{X|Y=0}$ is an unnormalized version of $\tilde{q}_X$, while what is left of $q'_{XY}$ in $q'_{X|Y=1}$ is such that its entropy, in conjunction with a bound on the entropy of $p_X$, yields the entropy of $\tilde{p}_X$.

Let $\sy \coloneqq \{0,1\}$. Let $\mathcal{P}_\sy$ denote the set of probability distributions on $\sy$ and $\mathcal{P}_{\sx \times \sy}$ denote the set of probability distributions on $\sx \times \sy$. For both the cases $d_+<|I|$ and $d_+ = |I|$, define the joint probability distribution $q'_{XY} \in \mathcal{P}_{\sx \times \sy}$ whose elements satisfy the following relations:
\begin{equation} \nonumber
	q'_{XY}(x,y) \coloneqq \begin{cases}
		\nu, & \textrm{if } y=0, \forall x \in \{ 1, \cdots, d_+ \}, \\
		\mu, & \textrm{if } y=0 \textrm{ and } x = d_+ + 1, \\
		0, & \textrm{if } y=0, \forall x \notin \{ 1, \cdots, d_+ + 1 \}, \\
		p_{X}(x), & \textrm{if } y=1, \forall x \in I, \\
		q_{X}(x), & \textrm{if } y=1, \forall x \in I^c.
	\end{cases}
\end{equation}
We indeed have $\sum_{y \in \sy} q'_{XY}(x,y) = q'_X(x)$,
so that $q'_{XY}$ is a genuine probability distribution with $\sx$-marginal $q'_X$. Now,
\begin{equation} \label{eq:StepCHq}
    \begin{aligned}
        \hr{X}{q'} & \stackrel{(a)}{\geq} \pr{X}{Y}{q'} \\
    & = \sum_{y \in \sy} q'_Y(y) \hr{X|Y=y}{q'} \\
    & \stackrel{(b)}{=} \eps \hr{X|Y=0}{q'} + (1-\eps) \hr{X|Y=1}{q'}.
    \end{aligned}
\end{equation}
In the above series of inequalities, $(a)$ follows from the fact that conditioning reduces entropy, while $(b)$ follows from $q'_Y(0) = \sum_x q'_{XY}(x,0) = \eps$, which also means that $q'_Y(1) = \sum_x q'_{XY}(x,1)= 1-\eps$.
\medskip

\noindent \textbf{Step D:} Define the function $f : [0,1] \rightarrow [0,\infty)$ as $f(u) \coloneqq - u \log u$. The concavity of the function $f$ and the fact that $f(0)=0$ {(which is justified by a continuity argument, since $\lim_{u \to 0} u \log u = 0$)}, implies that $f$ is subadditive.
Using this,
\begin{equation} \label{eq:ineqhp1}
    \begin{aligned}
        & \hr{X}{p} = \sum_{x \in \sx} f(p_X(x)) \\
        & \leq \sum_{x \in I} f(p_X(x)) + \sum_{x \in I^c} f(q_X(x)) + \sum_{x \in I^c} f(p_X(x)-q_X(x)).
    \end{aligned}
\end{equation}
As in the statement of Theorem~\ref{theo:main}, define $d_- \coloneqq |\sx| - \lceil \eps/\nu \rceil$.
As mentioned already, $d_+ \leq |I|$. If $\mu=0$, we then have that
\begin{equation}
	d_- = |\sx|-d_+ = |I^c|+(|I|-d_+) \geq |I^c|.
\end{equation}
We also showed that if $d_+ = |I|$, we necessarily have $\mu = 0$. This means that if $\mu \neq 0$, we have $d_+ \neq |I|$, which means that $d_+ \leq |I|-1$, or $|I| \geq d_++1$. In that case, when $\mu \neq 0$,
\begin{equation}
	d_- = |\sx|-d_+-1 = |I^c| + (|I|-d_+-1) \geq |I^c|.
\end{equation}
Consequently, irrespective of the value of $\mu \in [0,\nu)$, we always have that $d_- \geq |I^c|$. Now, $\sum_{x \in I^c} (p_X(x)-q_X(x)) = \eps$,
so that
\begin{equation} \nonumber
	\Big( \frac{\eps}{d_-}, \cdots, \frac{\eps}{d_-} \Big) \! \prec \! \Big( \frac{\eps}{|I^c|}, \cdots, \frac{\eps}{|I^c|} \Big) \! \prec \! \{ p_X(x)-q_X(x) \}_{x \in I^c},
\end{equation}
where the first vector in the above equation has $d_-$ elements, while the other two have $|I^c|$ elements, and when comparing vectors of different dimensions using a majorization relation, we append zeros to the vector with the smaller dimension. Since for any non-zero integer $n$, the function $F : \mathbb{R}^n \rightarrow \mathbb{R} : \bos{v} \mapsto \sum_{i=1}^n f(v_i)$
is symmetric and concave, it is also Schur concave, so that
\begin{equation} \nonumber
	\sum_{x \in I^c} f(p_X(x)-q_X(x)) \leq \sum_{x=1}^{d_-} f \left( \frac{\eps}{d_-} \right) = - \eps \log \left( \eps/d_- \right).
\end{equation}
Combining the above with~\eqref{eq:ineqhp1}, we obtain
\begin{equation} \label{eq:StepDHp}
	\hr{X}{p} \leq \sum_{x \in I} f(p_X(x)) + \sum_{x \in I^c} f(q_X(x)) - \eps \log \left( \eps/d_- \right).
\end{equation}

\noindent \textbf{Step E:} Substituting the expressions of $\hr{X|Y=0}{q'}$ and $\hr{X|Y=1}{q'}$
into the right-hand side of~\eqref{eq:StepCHq}, we obtain
\begin{align}
        \hr{X}{q'} & \geq - d_+ \nu \log (\nu/\eps) - \mu \log(\mu/\eps) + (1-\eps) \log (1-\eps) \nonumber \\
        & \quad \quad + \sum_{x \in I} f(p_X(x)) + \sum_{x \in I^c} f(q_X(x)), \label{eq:StepEHq}
\end{align}
where we used the fact that
\begin{equation}
	\sum_{x \in I} p_X(x) + \sum_{x \in I^c} q_X(x) = q'_Y(1) = 1-\eps.
\end{equation}
Inserting~\eqref{eq:StepEHq} and~\eqref{eq:StepDHp} into~\eqref{eq:StepBDeltaHi} (or~\eqref{eq:StepBDeltaHii}, depending on the case), and using~\eqref{eq:rem}, we obtain the required bound
\begin{equation} \label{eq:theo:ContinuityLocal3}
    \Delta H \leq h(\eps) + \eps \log d_- + d_+ \nu \log \nu + \mu \log \mu - \eps \log \eps.
\end{equation} 

In order to see that the bound is tight, one may verify that the bound is reached by the probability distributions $\tilde{q}_X$ and $\tilde{p}_X$, whose elements satisfy the following relations:
\begin{equation} \label{eq:probaTightMe1}
	\tilde{q}_{X}(x) \coloneqq \begin{cases}
		\nu/\eps, & \forall x \in \{ 1, \cdots, d_+ \}, \\
		\mu/\eps, & \textrm{if } x = d_+ + 1, \\
		0, & \textrm{else},
	\end{cases}
\end{equation}
and, when $\mu \neq 0$,
\begin{equation} \label{eq:probaTightMe2b}
	\tilde{p}_{X}(x) \coloneqq \begin{cases}
		(1-\eps)\nu/\eps, & \forall x \in \{ 1, \cdots, d_+ \}, \\
		(1-\eps)\mu/\eps, & \textrm{if } x = d_+ + 1, \\
		\eps/d_-, & \textrm{else},
	\end{cases}
\end{equation}
whereas, when $\mu = 0$,
\begin{equation} \label{eq:probaTightMe2a}
	\tilde{p}_{X}(x) \coloneqq \begin{cases}
		(1-\eps)\nu/\eps, & \forall x \in \{ 1, \cdots, d_+ \}, \\
		\eps/d_-, & \textrm{else}.
	\end{cases}
\end{equation}
By definition of the parameters $d_+$, $\mu$ and $d_-$, we see that all elements of $\tilde{q}_X$, and $\tilde{p}_X$, are non-negative and sum to one. Hence, $\tilde{q}_X$ and $\tilde{p}_X$ are bonafide probability distributions. It is also easy to see that $|\tilde{p}_{X}(x)-\tilde{q}_{X}(x)| \leq \nu$, for all $x \in \{ 1, \cdots, d_+ \}$ when $\mu=0$, and for all $x \in \{ 1, \cdots, d_+ +1 \}$ when $\mu>0$. Moreover, the inequality $\eps/d_- \leq \nu$ is a simple consequence of Lemma~\ref{lem:constDimQ} (taking $\rho$ and $\sigma$ in the latter to be commuting states and having eigenvalues given by the probability vectors $\tilde{p}_{X}$ and $\tilde{q}_{X}$, respectively). From this, it is easy to see that $\mathrm{LO}(\tilde{p}_X,\tilde{q}_X) = \nu$. We also have that $\mathrm{TV}(\tilde{p}_X,\tilde{q}_X) = \eps$. Finally, it is easy to see that $\left| \hr{X}{\tilde{p}} - \hr{X}{\tilde{q}} \right| = \hr{X}{\tilde{p}} - \hr{X}{\tilde{q}}$ saturates the bound both when $\mu=0$ and $\mu >0$. This ends the proof.

\subsection{Proof of Theorem~\ref{theo:ContinuityOpNorm} \label{sec:proof:theo:ContinuityOpNorm}}

Let the spectral decompositions of $\rho$ and $\sigma$ be given by
\begin{equation} \nonumber
    \rho = \sum_{x \in \sx} r_X(x) \proj{\phi_x}, \qquad \sigma = \sum_{x \in \sx} s_X(x) \proj{\psi_x},
\end{equation}
where $r_X, s_X \in \mathcal{P}_{\sx}$ and $|\sx|=d$.
Suppose, without loss of generality, that $S(\rho) \leq S(\sigma)$. Define the {\em{dephasing channel}} in the eigenbasis of $\rho$ as follows\footnote{Note that this is just a pinching map.}: for any state $\omega$,
\begin{equation}
	\mathcal{N}_d(\omega) \coloneqq \sum_{x=1}^d \proj{\phi_x} \omega \proj{\phi_x},
\end{equation}
We have that $\mathcal{N}_d(\rho) = \rho$. Furthermore, $\mathcal{N}_d$ is a unital channel, so that $\mathcal{N}_d(\sigma) \prec \sigma$ {(see \textit{e.g.}~\cite{Alberti1982})}, which implies $S(\mathcal{N}_d(\sigma)) \geq S(\sigma)$. Let the spectral decomposition of $\mathcal{N}_d(\sigma)$ be
\begin{equation}
    \mathcal{N}_d(\sigma) = \sum_{x \in \sx} \tilde{s}_X(x) \proj{\phi_x}.
\end{equation}
Putting all the previous facts together, we have
\begin{equation}
    H(X)_r = S(\rho) \leq S(\sigma) \leq S(\mathcal{N}_d(\sigma)) = H(X)_{\tilde{s}}.
\end{equation}
Meanwhile, from the data-processing inequality for the trace distance,
\begin{equation} \nonumber
    \eps \geq \eps_1 \coloneqq \mathrm{T}(\rho, \sigma) \geq \mathrm{T}(\mathcal{N}_d(\rho), \mathcal{N}_d(\sigma)) = \mathrm{TV}(r_X,\tilde{s}_X) \eqqcolon \eps_2.
\end{equation}
From the data-processing inequality for Schatten $p$-norms ($1 \leq p \leq \infty$) under unital channels (see Theorem~II.4 in~\cite{PerezGarcia2006})\footnote{Note that the same conclusion can be reached by exploiting the fact that the dephasing channel is a random unitary and using the triangle inequality.},
\begin{equation} \nonumber
    \begin{aligned}
        \nu \geq \nu_1 \coloneqq ||\rho-\sigma||_{\infty} & \geq ||\mathcal{N}_d(\rho)-\mathcal{N}_d(\sigma)||_{\infty} \\
        & = \mathrm{LO}(r_X,\tilde{s}_X) \eqqcolon \nu_2.
    \end{aligned}
\end{equation}
Now, using Lemma~\ref{lem:Sason}, we have
\begin{equation} \label{eq:prooftheo:ContinuityOpNorm1}
    \begin{aligned}
        \left| S(\rho) - S(\sigma) \right| & = S(\sigma) - S(\rho) \leq H(X)_{\tilde{s}} - H(X)_r \\
        & \leq \eps_2 \log\left((\nu_2/\eps_2) d-1\right) + h(\eps_2).
    \end{aligned}
\end{equation}
From Lemma~\ref{lem:constDimQ}, $d \geq 2 \left\lceil \frac{\eps_2}{\nu_2} \right\rceil$, so that $\nu \geq \nu_1 \geq \nu_2 \geq \eps_2/d$. Therefore, combining~\eqref{eq:prooftheo:ContinuityOpNorm1} with 
{Lemma~\ref{lem:incBoundSason1} in Appendix~\ref{sec:appA},}
we get
\begin{equation} \label{eq:prooftheo:ContinuityOpNorm2}
    \left| S(\rho) - S(\sigma) \right| \leq \eps_2 \log\left((\nu/\eps_2 )d-1\right) + h(\eps_2).
\end{equation}
Finally, from the hypothesis of Theorem~\ref{theo:ContinuityOpNorm}, $\eps_2 \leq \eps_1 \leq \eps \leq \nu d/(\nu d + 3)$. Therefore, combining~\eqref{eq:prooftheo:ContinuityOpNorm2} with 
{Lemma~\ref{lem:incBoundSason2} in Appendix~\ref{sec:appA},}
we end up with the desired bound
\begin{equation} \label{eq:prooftheo:ContinuityOpNorm3}
    \left| S(\rho) - S(\sigma) \right| \leq \eps \log\left((\nu/\eps) d-1\right) + h(\eps).
\end{equation}
To show that the continuity bound is tight when $1/\beta$ is an integer, one can simply consider commuting states $\rho$ and $\sigma$ whose eigenvalues are given by the probability distributions of~\eqref{eq:probaTightSason1} and~\eqref{eq:probaTightSason2}, respectively, with $|\sx|=d$.

\section{\texorpdfstring{$(\eps,\nu)$--degradable channels and upper bounds on their quantum and private classical capacities}{(epsilon,nu)--degradable channels and upper bounds on their quantum and private classical capacities}}
\label{sec:quantum-cap}

The {\em{quantum capacity}} of a quantum channel is the maximum rate at which quantum information can be transmitted through it reliably (in the asymptotic limit\footnote{That is, in the limit $n \to \infty$, where $n$ denotes the number of independent and identical uses of the channel.}) per use of the channel. For a quantum channel $\Phi: A \to B$, the quantum capacity is given by the so-called LSD theorem, named after Lloyd, Shor and Devetak, whose independent works~\cite{Lloyd1997,Shor2002,Devetak2005}
 led to it:
\begin{equation}\label{eq:qcap}
    Q(\Phi) = \lim_{k \rightarrow \infty} \frac{1}{k} Q^{(1)}(\Phi^{\otimes k}),
\end{equation}
where $Q^{(1)}(\Phi)$ denotes the {\em{channel coherent information}} of $\Phi$ and is given by
\begin{equation} \label{eq-1}
    Q^{(1)}(\Phi) \coloneqq \max_{\rho \in \mathcal{D}(\mathcal{H}_A)} \left( S(\Phi(\rho)) - S(\Phi^{\mathrm{c}}(\rho)) \right).
\end{equation}

The {\em{private classical capacity}} of a quantum channel is the maximum rate at which classical information can be transmitted through it reliably (in the asymptotic limit) per use of the channel, such that no information about the message leaks to an eavesdropper who might have access to the environment. It was shown independently by Cai et al~\cite{Cai2004} and Devetak~\cite{Devetak2005} that
\begin{equation}\label{eq:pcap}
    P(\Phi) = \lim_{k \rightarrow \infty} \frac{1}{k} P^{(1)}(\Phi^{\otimes k}),
\end{equation}
where $P^{(1)}(\Phi)$ denotes the {\em{channel private information}} of $\Phi$ and is given by
\begin{equation} \nonumber
    \begin{aligned}
        P^{(1)}(\Phi) \coloneqq \max_{\{p_i,\rho_i\}} & \Biggl( S\biggl(\sum_i p_i \Phi(\rho_i)\biggr) - \sum_i p_i S\left( \Phi(\rho_i)\right) \\
        & - S\biggl(\sum_i p_i \Phi^{\mathrm{c}}(\rho_i)\biggr) + \sum_i p_i S\left( \Phi^{\mathrm{c}}(\rho_i)\right) \Biggr),
    \end{aligned}
\end{equation}
where the maximization is over ensembles of quantum states $\{p_i, \rho_i\}$.
The channel coherent information, $Q^{(1)}(\Phi)$, {and the channel private information, $P^{(1)}(\Phi)$,} are in general neither additive nor convex~\cite{Smith-Yard-Science,Smith2008b}. Hence the regularizations in~\eqref{eq:qcap} {and~\eqref{eq:pcap} are necessary and this makes the quantum and private classical capacity} difficult to compute. There is, however, an important family of channels for which the channel coherent information and the private information are additive, and hence the quantum and private classical capacities are given by single letter formulae. This is the family of {\em{degradable channels}}. A quantum channel $\Phi: A \to B$ is said to be {\em{degradable}} if it can be {\em{degraded}} to its complementary channel $\Phi^c: A \to E$, i.e., there exists a linear CPTP map (quantum channel) such that $\Phi^c = \Lambda \circ \Phi$, where $\circ$ denotes composition of channels. Such a map $\Lambda$ is often referred to as a {\em{degrading channel}} for the quantum channel $\Phi$.

The notion of degradability of a quantum channel is, however, not robust, in the sense that a small perturbation of a degradable channel can destroy its property of degradability. In~\cite{Sutter2017}, Sutter et al. introduced the notion of {\em{approximate degradable quantum channels}}. These are channels for which the degradability condition is approximately satisfied: 
\begin{defi}[Definition~4 of~\cite{Sutter2017}, $\eps$--degradable]
    A quantum channel $\Phi : A \to B$ is said to be $\eps$--degradable, for some $\eps \geq 0$, if there exists a (degrading) quantum channel $\Lambda : B \to E$ such that $||\Phi^{\mathrm{c}} - \Lambda \circ \Phi ||_\dia \leq \eps$.
\end{defi}
The quantum capacity, $Q(\Phi)$ of any channel $\Phi$ is trivially lower-bounded by the channel coherent information $Q^{(1)}(\Phi)$. However, as proved in~\cite{Sutter2017}, for $\eps$--degradable channels, the quantum capacity is also bounded from above by $Q^{(1)}(\Phi)$, up to error terms which vanish as $\eps$ goes to zero. This is because $\eps$--degradable channels, approximately inherit the desirable additivity properties of the channel coherent information of degradable channels (see~$(i)$ in Theorem~7 of~\cite{Sutter2017}). {Similarly, the private classical capacity, $P(\Phi)$ of any channel $\Phi$ is lower-bounded by the channel private information $P^{(1)}(\Phi)$. For $\eps$--degradable channels, it is also bounded from above by $P^{(1)}(\Phi)$, up to error terms which vanish as $\eps$ goes to zero (see~$(iii)$ in Theorem~7 of~\cite{Sutter2017}).}


\subsection{\texorpdfstring{$(\eps,\nu)$--degradable channels}{(Epsilon,Nu)--degradable channels}}

In this section, we refine the notion of approximate degradable channels by introducing the definition of $(\eps,\nu)$--degradable channels: these are a subset of the set of $\eps$--degradable channels, for which the operator norm distance between the complementary channel and the channel composed with a degrading map is at most $\nu$. We then proceed to prove that for such channels, one can get a tighter upper bound on the quantum capacity than that obtained in~\cite{Sutter2017}.

\begin{defi}[$(\eps,\nu)$--degradable] \label{defi:EpsNuDegCHannel}
     A quantum channel $\Phi : A \to B$ is said to be $(\eps,\nu)$--degradable, for some $\eps \geq 0$ and some $\nu \geq 0$, if there exists a (degrading) quantum channel $\Lambda : B \to E$ such that $||\Phi^{\mathrm{c}} - \Lambda \circ \Phi ||_\dia \leq \eps$ and $||\Phi^{\mathrm{c}} - \Lambda \circ \Phi ||_\cb \leq \nu$.
\end{defi}

\noindent 
\begin{rem}By Definition~\ref{defi:EpsNuDegCHannel}, every channel is $(\eps,\nu)$--degradable for some $\eps \in [0,2]$ and $\nu$ in $[0,2 d_Ad_E]$, where $d_A:=\mathrm{dim} \, \mathcal{H}_A$ and $d_E:=\mathrm{dim} \, \mathcal{H}_E$.
Indeed, $||\Phi^{\mathrm{c}} - \Lambda \circ \Phi||_\dia \leq ||\Phi^{\mathrm{c}}||_\dia + ||\Lambda \circ \Phi||_\dia \leq 2$, while $||\Phi^{\mathrm{c}} - \Lambda \circ \Phi||_\cb \leq ||\Phi^{\mathrm{c}}||_\cb + ||\Lambda \circ \Phi||_\cb \leq 2 d_Ad_E$. Here, we used that fact that, for any CPTP map $\Psi : A \rightarrow E$,
\begin{equation}
\|\Psi\|_{\cb} = \| \Psi \otimes {\rm{id}}_{E}\|_{\infty \rightarrow \infty} \leq d_Ad_E,
\end{equation}
where the inequality is a consequence of Theorem~II.1 in~\cite{PerezGarcia2006}.
\end{rem}
It follows from the definition that if a channel $\Phi$ is $(\eps,\nu)$--degradable, it is also $(\eps',\nu')$--degradable for all $\eps' \geq \eps$ and $\nu' \geq \nu$. 
Fix $\nu \in [0,2 d_Ad_E]$. For such a $\nu$,
the smallest possible parameter $\eps$ such that $\Phi$ is $(\eps,\nu)$--degradable is given by
\begin{equation} \label{eq:optepsGen}
    \eps_\Phi \coloneqq \inf_{\Lambda} \left\{ || \Phi^{\mathrm{c}} - \Lambda \circ \Phi ||_{\dia} : \Lambda : \mathcal{L}(B) \rightarrow \mathcal{L}(E) \, \text{is CPTP} \right\}.
\end{equation}
It was shown in~\cite{Sutter2017} that the above optimization problem can be expressed as a semidefinite program (SDP) (see Proposition~9 in~\cite{Sutter2017}). Similarly,
fix $\eps \in [0, 2]$. For such an $\eps$,
the smallest possible parameter $\nu$ such that $\Phi$ is $(\eps,\nu)$--degradable is given by
\begin{equation} \label{eq:optnuGen}
    \nu_\Phi \coloneqq \inf_{\Lambda} \left\{ || \Phi^{\mathrm{c}} - \Lambda \circ \Phi ||_{\cb} : \Lambda : \mathcal{L}(B) \rightarrow \mathcal{L}(E) \, \text{is CPTP} \right\}.
\end{equation}
\begin{prop}
    The optimization problem~\eqref{eq:optnuGen} can be expressed as a semidefinite program.
\end{prop}
\begin{proof}
    From~\eqref{eq:dualityDiamondCB}, we have that
    \begin{equation}
        || \Phi^{\mathrm{c}} - \Lambda \circ \Phi ||_{\cb} = || (\Phi^{\mathrm{c}})^* - (\Lambda \circ \Phi)^* ||_{\dia},
    \end{equation}
    so that we can write~\eqref{eq:optnuGen} as
    \begin{equation} \nonumber
        \nu_\Phi = \inf_{\Lambda} \left\{ || (\Phi^{\mathrm{c}} - \Lambda \circ \Phi)^* ||_{\dia} : \Lambda : \mathcal{L}(B) \rightarrow \mathcal{L}(E) \, \text{is CPTP} \right\}.
    \end{equation}
    or,
    \begin{equation} \label{eq:optnu3}
        \nu_\Phi = \inf_{\Theta} \left\{ || (\Phi^{\mathrm{c}})^* - \Phi^* \circ \Theta ||_{\dia} : \Theta : \mathcal{L}(E) \rightarrow \mathcal{L}(B) \, \text{is CPU} \right\}.
    \end{equation}
    We now follow the same steps as in the proof of Proposition~9 in~\cite{Sutter2017} to show that the optimization problem~\eqref{eq:optnu3} can be expressed as an SDP. Watrous showed~\cite{Watrous2009} that for two completely positive maps $\Gamma_1, \Gamma_2 : \mathcal{L}(E) \rightarrow \mathcal{L}(A)$, the diamond norm of their difference can be expressed as an SDP of the form
    \begin{equation}
        \begin{aligned}
            || \Gamma_1 - \Gamma_2 ||_{\dia} = \, & 2 \inf_Z || \Tr_A(Z) ||_{\infty} \\
            & \text{s. t.} \quad Z \geq J(\Gamma_1 - \Gamma_2) \\
            & \phantom{\text{s. t.} \quad} Z \geq 0,
        \end{aligned}
    \end{equation}
    where $J(\Gamma_1 - \Gamma_2)$ denotes the Choi operator of $(\Gamma_1 - \Gamma_2)$. Using this, \eqref{eq:optnu3} becomes
    \begin{equation} \label{eq:optnu4}
        \begin{aligned}
            \nu_\Phi = \, & 2 \inf_{\Theta} \inf_Z || \Tr_A(Z) ||_{\infty} \\
            & \text{s. t.} \quad Z \geq J((\Phi^{\mathrm{c}})^*) - J(\Phi^* \circ \Theta) \\
            & \phantom{\text{s. t.} \quad} Z \geq 0 \\
            & \phantom{\text{s. t.} \quad} J(\Theta) \geq 0 \\
            & \phantom{\text{s. t.} \quad} \Tr_E (J(\Theta)) = \mathds{1}_B,
        \end{aligned}
    \end{equation}
    where the two final constraints in~\eqref{eq:optnu4} ensure that $\Theta$ is CPU. Since the two minimizations can be interchanged, they can be merged, so that~\eqref{eq:optnu3} becomes
    \begin{equation} \label{eq:optnu5}
        \begin{aligned}
            \nu_\Phi = \, & 2 \inf_{Z,\Theta} || \Tr_A(Z) ||_{\infty} \\
            & \text{s. t.} \quad Z \geq J((\Phi^{\mathrm{c}})^*) - J(\Phi \circ \Theta) \\
            & \phantom{\text{s. t.} \quad} Z \geq 0 \\
            & \phantom{\text{s. t.} \quad} J(\Theta) \geq 0 \\
            & \phantom{\text{s. t.} \quad} \Tr_E (J(\Theta)) = \mathds{1}_B,
        \end{aligned}
    \end{equation}
    which is indeed an SDP.
\end{proof}


\subsection{Upper bounds on quantum {and private classical} capacities}

For a quantum channel $\Phi: A \to B$ and a degrading channel $\Lambda: B \to \tilde{E}$, with ${\mathcal{H}}_{\tilde{E}} \simeq {\mathcal{H}}_E$, choose Stinespring isometries $V : A \to B \otimes E$ and $W : B \to \tilde{E} \otimes F$, respectively. Then, define~\cite{Sutter2017}
\begin{equation} \label{eq:maxCondEnt}
    \begin{aligned}
            S(\Phi,\Lambda) \coloneqq \max_{\rho \in \mathcal{D}(\mathcal{H}_A)} \{ & S(F|\tilde{E})_{\omega} : \\ & \omega_{E\tilde{E}F} = (W \otimes \mathds{1}) V \rho V^{\dagger} (W^\dagger \otimes \mathds{1})\},
    \end{aligned}
\end{equation}
where $S(F|\tilde{E})_{\omega} \coloneqq S(\omega_{\tilde{E}F}) - S(\omega_{\tilde{E}})  $ is the conditional entropy, with $\omega_{\tilde{E}F} = \Tr_E(\omega_{E\tilde{E}F})$ and $\omega_{\tilde{E}} = \Tr_{EF}(\omega_{E\tilde{E}F})$.
Note that $S(F|\tilde{E})_{\omega} = S(\Phi(\rho))-S(\Lambda \circ \Phi(\rho))$.

\begin{rem}
    Note that even though we previously denoted the degrading map as $\Lambda : B \to E$, here we choose to denote it as $\Lambda : B \to \tilde{E}$, with the Hilbert spaces of $E$ and $\tilde{E}$ being isomorphic, because of the following reasons: $(i)$ we need to consider the tripartite state $\omega_{E\tilde{E}F}$ resulting from the isometries $V$ and $W$ defined above; $(ii)$ its reduced states $\omega_E$ and $\omega_{\tilde{E}}$ are not identical, since the channel is only approximately degradable.
\end{rem}

Sutter et al.~\cite{Sutter2017} showed, by employing the Audenaert-Fannes-Petz inequality (Theorem~\ref{th-2}) that if $\Phi : A \to B$ is an $\eps$--degradable channel with a degrading channel $\Lambda : B \to E$ then (see Proposition~5 of~\cite{Sutter2017})
\begin{equation}
    |Q^{(1)}(\Phi)-S(\Phi,\Lambda)| \leq \frac{\eps}{2} \log (d_E-1) + h(\eps/2),
\end{equation}
where $d_E = {\rm{dim}} \,{\mathcal{H}}_E$. 
Furthermore, they showed, by employing the so-called Alicki-Fannes-Winter inequality~\cite{AlickiFannes,Winter2016} for conditional entropy the following upper bound on the quantum capacity of an $\eps$--degradable channel (see inequality~$(ii)$ in Theorem~7 of~\cite{Sutter2017}):
\begin{equation} \label{eq:boundCapaAlickyFannesSutter}
    Q(\Phi) \leq S(\Phi,\Lambda) + \eps \log(d_E) + g(\eps/2).
\end{equation}
Since any $(\eps,\nu)$--degradable channel is, by definition, also $\eps$--degradable, the bound~\eqref{eq:boundCapaAlickyFannesSutter} also holds for the quantum capacity of a $(\eps,\nu)$--degradable channel. From the above, one can get the following upper bound on the quantum capacity of an $\eps$--degradable channel
\begin{equation}
    \begin{aligned}
        Q(\Phi) \leq Q^{(1)}(\Phi)
        & + \frac{\eps}{2} \log(d_E - 1) + h(\eps/2) \\
        & \quad + \eps \log(d_E) + g(\eps/2).
    \end{aligned}
\end{equation}

We show that the above upper bound can be strengthened by exploiting the notion of $(\eps,\nu)$--degradable channels as follows.
\begin{theo} \label{theo:cohinfcont}
 Let $\Phi : A \to B$ be an $(\eps,\nu)$--degradable channel with a degrading channel $\Lambda : B \to E$, with ${\rm{dim}}\, {\mathcal{H}}_E = d_E$. 
 If $\eps \leq 2\nu d_E/(\nu d_E + 3)$, then for $\beta \coloneqq 2\nu/\eps$,
    \begin{equation} \label{eq:theo:cohinfcont}
        |Q^{(1)}(\Phi)-S(\Phi,\Lambda)| \leq \frac{\eps}{2} \log(\beta d_E-1) + h(\eps/2).
    \end{equation}
\end{theo}
\noindent This leads to the following upper bound on the quantum capacity of an $(\eps,\nu)$--degradable channel.
\begin{theo} \label{theo:upperQuantCapa}
Let $\Phi : A \to B$ be an $(\eps,\nu)$--degradable channel with a degrading channel $\Lambda : B \to E$, with ${\rm{dim}}\, {\mathcal{H}}_E = d_E$.
If $\eps \leq 2\nu d_E/(\nu d_E + 3)$, then for $\beta \coloneqq 2\nu/\eps$,
    \begin{equation} \label{eq:theo:upperQuantCapa}
        \begin{aligned}
            Q(\Phi) \leq Q^{(1)}(\Phi)
        & + \frac{\eps}{2} \log(\beta d_E-1) + h(\eps/2) \\
        & \quad + \eps \log(d_E) + g(\eps/2).
        \end{aligned}
    \end{equation}
\end{theo}
\begin{proof}[Proof of Theorem~\ref{theo:cohinfcont}]
    Our proof is analogous to that of Proposition~5 in~\cite{Sutter2017}. The channel coherent information can be written as
    \begin{equation} \nonumber
        \begin{aligned}
            Q^{(1)}(\Phi) = \max_{\rho \in \mathcal{D}(\mathcal{H}_A)} \{ & S(\omega_{\tilde{E}F}) - S(\omega_E) : \\
            & \omega_{E\tilde{E}F} = (W \otimes \mathds{1}) V \rho V^{\dagger} (W^\dagger \otimes \mathds{1}) \},
        \end{aligned}
    \end{equation}
    with $\omega_{E} = \Tr_{\tilde{E}F}(\omega_{E\tilde{E}F})$ (see~\eqref{eq:maxCondEnt}). Since $\Phi$ is $(\eps,\nu)$--degradable, it can be seen that $||\omega_E-\omega_{\tilde{E}}||_1 \leq \eps$ (and hence $T(\omega_E,\omega_{\tilde{E}}) \leq \eps/2)$) and $||\omega_E-\omega_{\tilde{E}}||_\infty \leq \nu$.
   Indeed, defining ${\mathcal{H}}_{E'} \simeq {\mathcal{H}}_E$,
    \begin{equation} \nonumber
        \begin{aligned}
            & ||\Phi^{\mathrm{c}} - \Lambda \circ \Phi ||_\cb \leq \nu \\
            \Rightarrow \quad & ||[(\Phi^{\mathrm{c}} - \Lambda \circ \Phi) \otimes {\rm{id}}_{E'}](\rho \otimes \sigma)||_\infty \leq \nu, \\
            & \quad \quad \forall \, \rho \in \mathcal{L}(\mathcal{H}_A), \sigma \in \mathcal{L}(\mathcal{H}_{E'}) \, \text{ s.t. } \, ||\rho||_\infty, ||\sigma||_\infty \leq 1, \\
            \stackrel{(a)}{\Rightarrow} \quad &  ||(\omega_E-\omega_{\tilde{E}}) \otimes \sigma||_\infty \leq \nu, \forall \, \sigma \in \mathcal{L}(\mathcal{H}_{E'}) \, \text{ s.t. } \, ||\sigma||_\infty \leq 1, \\
            \stackrel{(b)}{\Rightarrow} \quad &  ||(\omega_E-\omega_{\tilde{E}})||_\infty \leq \nu,
        \end{aligned}
    \end{equation}
    where we considered a density matrix $\rho$ satisfying $\omega_{E\tilde{E}F} = (W \otimes \mathds{1}) V \rho V^{\dagger} (W^\dagger \otimes \mathds{1})$ (as in~\eqref{eq:maxCondEnt}) in $(a)$ and took any density matrix $\sigma$ satisfying $||\sigma||_\infty = 1$ in $(b)$ ($\rho$ and $\sigma$ therefore both satisfy $||\rho||_\infty, ||\sigma||_\infty \leq 1$). Since $\eps \leq 2\nu d_E/(\nu d_E + 3)$, we can therefore apply Theorem~\ref{theo:ContinuityOpNorm} to replace $S(\omega_E)$ with $S(\omega_{\tilde{E}}) - \left(\eps/2 \log(\beta d_E - 1) + h({\eps}/{2})\right)$, and the claim follows.
\end{proof}

\begin{proof}[Proof of Theorem~\ref{theo:upperQuantCapa}]
Since $\Phi$ is an $(\eps,\nu)$--degradable channel with a degrading channel $\Lambda$, with $\eps \leq 2\nu d_E/(\nu d_E + 3)$, we have from Theorem~\ref{theo:cohinfcont} that
\begin{equation}
    S(\Phi,\Lambda) \leq Q^{(1)}(\Phi) + \frac{\eps}{2} \log(\beta d_E-1) + h(\eps/2),
\end{equation}
since the right-hand side of~\eqref{eq:theo:cohinfcont} is non-negative. Combining the above relations with~\eqref{eq:boundCapaAlickyFannesSutter} gives~\eqref{eq:theo:upperQuantCapa}.
\end{proof}

The upper bound on the quantum capacity given in Theorem~\ref{theo:upperQuantCapa} pertain to $(\eps,\nu)$--degradable channels, whose definition involves the diamond norm and the cb norm (see Definition~\ref{defi:EpsNuDegCHannel}). However, it is sometimes possible to obtain a better upper bound on the quantum capacity of a quantum channel in terms of efficiently computable upper bounds on unstabilised versions of these norms, with the inputs being restricted to density matrices (see~\eqref{eq:normUnstabilized}). Indeed, we obtain the following propositions, whose proofs are analogous to those of Theorems~\ref{theo:cohinfcont} and~\ref{theo:upperQuantCapa}.
\begin{prop} \label{prop:cohinfcont2}
Let $\Phi : A \to B$ be a channel. If there exists a (degrading) channel $\Lambda : B \to E$ such that $\| \Phi^{\mathrm{c}} - \Lambda \circ \Phi \|_1^{\mathrm{D}} \leq \eps$ and $\| \Phi^{\mathrm{c}} - \Lambda \circ \Phi \|_{\infty}^{\mathrm{D}} \leq \nu$, with ${\rm{dim}}\, {\mathcal{H}}_E = d_E$ and $\eps \leq 2\nu d_E/(\nu d_E + 3)$, then for $\beta \coloneqq 2\nu/\eps$,
    \begin{equation} \label{eq:prop:cohinfcont2}
        |Q^{(1)}(\Phi)-S(\Phi,\Lambda)| \leq \frac{\eps}{2} \log(\beta d_E-1) + h(\eps/2).
    \end{equation}
\end{prop}
\begin{prop} \label{prop:upperQuantCapa2}
    Let $\Phi : A \to B$ be a channel. If there exists a (degrading) channel $\Lambda : B \to E$ such that $\| \Phi^{\mathrm{c}} - \Lambda \circ \Phi \|_1^{\mathrm{D}} \leq \eps_1$, $\| \Phi^{\mathrm{c}} - \Lambda \circ \Phi \|_{\dia} \leq \eps_{\dia}$ and $\| \Phi^{\mathrm{c}} - \Lambda \circ \Phi \|_{\infty}^{\mathrm{D}} \leq \nu$, with ${\rm{dim}}\, {\mathcal{H}}_E = d_E$ and $\eps_1 \leq 2\nu d_E/(\nu d_E + 3)$, then for $\beta \coloneqq 2\nu/\eps_1$,
    \begin{equation} \label{eq:prop:upperQuantCapa2}
        \begin{aligned}
            Q(\Phi) \leq Q^{(1)}(\Phi)
        & + \frac{\eps_1}{2} \log(\beta d_E-1) + h(\eps_1/2) \\
        & \quad + \eps_{\dia} \log(d_E) + g(\eps_{\dia}/2).
        \end{aligned}
    \end{equation}
\end{prop}

\noindent The following simple Lemma allows us to make an even stronger case for why Proposition~\ref{prop:upperQuantCapa2} provides a better bound than Theorem~\ref{theo:upperQuantCapa}.
\begin{lem} \label{lem:1toInfLeq1to1}
Let $\Phi_1 : A \to B$ and $\Phi_2 : A \to B$ be two quantum channels, then
\begin{equation}
    \| \Phi_1 - \Phi_2 \|_{\infty}^{\mathrm{D}} \leq \frac{\| \Phi_1 - \Phi_2 \|_1^{\mathrm{D}}}{2}.
\end{equation}
\end{lem}
\begin{proof}[Proof of Lemma~\ref{lem:1toInfLeq1to1}]
    From Lemma~\ref{lem:constDimQ}, $\|\sigma_1-\sigma_2\|_{\infty} \leq \|\sigma_1-\sigma_2\|_1/2$ for any $\sigma_1, \sigma_2 \in \mathcal{D}(\mathcal{H}_B)$, so that $\|\Phi_1(\rho)-\Phi_2(\rho)\|_{\infty} \leq \|\Phi_1(\rho)-\Phi_2(\rho)\|_1/2$ for any $\rho \in \mathcal{D}(\mathcal{H}_A)$. Maximizing both sides over $\rho \in \mathcal{D}(\mathcal{H}_A)$ gives the result.
\end{proof}
\begin{rem} \label{rem:3}
If one can compute the quantities $\| \Phi^{\mathrm{c}} - \Lambda \circ \Phi \|_1^{\mathrm{D}}$, $\| \Phi^{\mathrm{c}} - \Lambda \circ \Phi \|_{\dia}$ and $\| \Phi^{\mathrm{c}} - \Lambda \circ \Phi \|_{\infty}^{\mathrm{D}}$ for any quantum channel $\Lambda$, one obtains the lowest upper bounds on $Q(\Phi)$ (given by~\eqref{eq:prop:upperQuantCapa2}) by taking $\eps_1 \coloneqq \| \Phi^{\mathrm{c}} - \Lambda \circ \Phi \|_1^{\mathrm{D}}$, $\eps_{\dia } \coloneqq \| \Phi^{\mathrm{c}} - \Lambda \circ \Phi \|_{\dia}$ and $\nu \coloneqq \| \Phi^{\mathrm{c}} - \Lambda \circ \Phi \|_{\infty}^{\mathrm{D}}$. In that case, Lemma~\ref{lem:1toInfLeq1to1} implies that $\beta \leq 1$ in Proposition~\ref{prop:upperQuantCapa2}. Coupled with the fact that $\eps_1 \leq \eps_\dia$ by definition, this guarantees that the bound given in~\eqref{eq:prop:upperQuantCapa2} (for these choices of $\eps_1$, $\eps_{\dia }$ and $\nu$) is at least as good as the bound given in point~$(i)$ of Theorem~7 in~\cite{Sutter2017}.\footnote{In contrast, the cb norm of the difference of two quantum channels is not necessarily bounded from above by half the diamond norm of the difference, meaning that $\beta$ appearing in Theorem~\ref{theo:upperQuantCapa} is not necessarily bounded from above by $1$ if one chooses $\eps \coloneqq \| \Phi^{\mathrm{c}} - \Lambda \circ \Phi \|_{\dia}$ and $\nu \coloneqq \| \Phi^{\mathrm{c}} - \Lambda \circ \Phi \|_{\cb}$. This is, for example, the case when $\Phi$ is the depolarizing channel. Indeed, later in Section~\ref{sec:depol}, we show that Proposition~\ref{prop:upperQuantCapa2} provides us with a better upper bound than Theorem~\ref{theo:upperQuantCapa} for the depolarizing channel.}
\end{rem}
\noindent Note that, for any quantum channels $\Phi$ and $\Lambda$, as in Proposition~\ref{prop:upperQuantCapa2}, one can always find finite non-negative parameters $\eps_1$, $\eps_\dia$ and $\nu$ such that $\| \Phi^{\mathrm{c}} - \Lambda \circ \Phi \|_1^{\mathrm{D}} \leq \eps_1$, $\| \Phi^{\mathrm{c}} - \Lambda \circ \Phi \|_{\dia} \leq \eps_{\dia}$ and $\| \Phi^{\mathrm{c}} - \Lambda \circ \Phi \|_{\infty}^{\mathrm{D}} \leq \nu$. Indeed, as already mentioned, $||\Phi^{\mathrm{c}} - \Lambda \circ \Phi||_\dia \leq 2$. Furthermore, $\| \Phi^{\mathrm{c}} - \Lambda \circ \Phi \|_1^{\mathrm{D}} \leq ||\Phi^{\mathrm{c}} - \Lambda \circ \Phi||_\dia \leq 2$ and $\| \Phi^{\mathrm{c}} - \Lambda \circ \Phi \|_{\infty}^{\mathrm{D}} \leq 1$ using Lemma~\ref{lem:1toInfLeq1to1}.
\bigskip

\noindent \textbf{Upper bounds on private classical capacities.} The upper bound on the private classical capacity of a quantum channel obtained by Sutter et al.~\cite{Sutter2017} (see point~$(iii)$ of Theorem~7) can be analogously strengthened by exploiting the notion of $(\eps,\nu)$--degradable channels. Its proof exploits the Audenaert-Fannes-Petz inequality~\eqref{eq:th-2} and the Alicki-Fannes-Winter inequality~\cite{AlickiFannes,Winter2016}. By replacing the former inequality by the inequality in our Theorem~\ref{theo:ContinuityOpNorm}, we obtain the following.
\begin{theo} \label{theo:upperPrivCapa}
Let $\Phi : A \to B$ be an $(\eps,\nu)$--degradable channel with a degrading channel $\Lambda : B \to E$, with ${\rm{dim}}\, {\mathcal{H}}_E = d_E$.
If $\eps \leq 2\nu d_E/(\nu d_E + 3)$, then for $\beta \coloneqq 2\nu/\eps$,
    \begin{equation} \nonumber
        \begin{aligned}
            P(\Phi) \leq P^{(1)}(\Phi)
            & + \frac{\eps}{2} \log(\beta d_E-1) + h(\eps/2) \\
            & \quad + 3 \eps \log(d_E) + 3 g(\eps/2).
        \end{aligned}
    \end{equation}
\end{theo}
\begin{prop} \label{prop:upperPrivCapa2}
    Let $\Phi : A \to B$ be a channel. If there exists a (degrading) channel $\Lambda : B \to E$ such that $\| \Phi^{\mathrm{c}} - \Lambda \circ \Phi \|_1^{\mathrm{D}} \leq \eps_1$, $\| \Phi^{\mathrm{c}} - \Lambda \circ \Phi \|_{\dia} \leq \eps_{\dia}$ and $\| \Phi^{\mathrm{c}} - \Lambda \circ \Phi \|_{\infty}^{\mathrm{D}} \leq \nu$, with ${\rm{dim}}\, {\mathcal{H}}_E = d_E$ and $\eps_1 \leq 2\nu d_E/(\nu d_E + 3)$, then for $\beta \coloneqq 2\nu/\eps_1$,
    \begin{equation} \label{eq:prop:upperPrivCapa2}
        \begin{aligned}
            P(\Phi) \leq P^{(1)}(\Phi)
            & + \frac{\eps_1}{2} \log(\beta d_E-1) + h(\eps_1/2) \\
            & \quad + 3 \eps_{\dia} \log(d_E) + 3 g(\eps_{\dia}/2).
        \end{aligned}
    \end{equation}
\end{prop}


\subsection{\texorpdfstring{Upper bounds on $\| \cdot \|_{\infty}^{\mathrm{D}}$ and $\| \cdot \|_{1}^{\mathrm{D}}$ for a difference of two quantum channels and corresponding SDPs}{Upper bound on 1 to infinity norm and 1 to 1 norm for a difference of two quantum channels and corresponding SDPs}}

The purpose of this section is to provide computable upper bounds on the unstabilised norms $\| \Phi \|_{\infty}^{\mathrm{D}}$ and $\| \Phi\|_{1}^{\mathrm{D}}$ (defined through~\eqref{eq:normUnstabilized}), where $\Phi$ is a difference of two quantum channels. These in turn allow us to compute the upper bounds on the quantum and private classical capacities given in Propositions~\ref{prop:upperQuantCapa2} and~\ref{prop:upperPrivCapa2}. As we show below, the upper bounds on the norms can be efficiently expressed as semidefinite programs (SDPs). For the sake of completeness we include a brief description of SDPs below.
\medskip

\noindent
\textbf{Semidefinite programming.} We adopt the formulation given by Watrous~\cite{bookWatrous} to define semidefinite programs (see also~\cite{Watrousln,Watrous2009,Watrous2013}). A semidefinite program (SDP) over two complex vector spaces $\mathcal{X} = \mathbb{C}^n$ and $\mathcal{Y} = \mathbb{C}^m$ is specified by a triple $(\Psi,C,D)$, where:
\begin{enumerate}
    \item $\Psi : \mathcal{L}(\mathcal{X}) \rightarrow \mathcal{L}(\mathcal{Y})$ is a Hermiticity-preserving super-operator;
    \item $C \in \mathcal{L}(\mathcal{X})$ and $D \in \mathcal{L}(\mathcal{Y})$ are Hermitian operators.
\end{enumerate}
The following two optimization problems are associated with such a semidefinite program:
\begin{center}
    \begin{minipage}{.49\columnwidth}
    \begin{center}
        \underline{Primal problem}
        \begin{equation*}
            \begin{aligned}
                \text{maximize:} \quad & \inner{C}{X}, \\
                \text{subject to:} \quad & \Psi(X) \leq D, \\
                & X \geq 0, \\
                & X \in \mathcal{L}(\mathcal{X}).
            \end{aligned}
        \end{equation*}
    \end{center}
    \end{minipage}
    \begin{minipage}{.49\columnwidth}
    \begin{center}
        \underline{Dual problem}
        \begin{equation*}
            \begin{aligned}
                \text{minimize:} \quad & \inner{D}{Y}, \\
                \text{subject to:} \quad & \Psi^*(Y) \geq C, \\
                & Y \geq 0, \\
                & Y \in \mathcal{L}(\mathcal{Y}).
            \end{aligned}
        \end{equation*}
    \end{center}
    \end{minipage}
\end{center}
Denote by $\mathscr{P}$ and $\mathscr{D}$ the sets of so-called primal and dual \emph{feasible operators}, respectively, i.e.,
\begin{equation}
    \begin{aligned}
        \mathscr{P} & = \{ X \in \mathcal{L}(\mathcal{X}) : X \geq 0, \Psi(X) \leq D \}, \\
        \mathscr{D} & = \{ Y \in \mathcal{L}(\mathcal{Y}) : Y \geq 0, \Psi^*(X) \geq C \}.
    \end{aligned}
\end{equation}
For every SDP, it holds that
\begin{equation}
    \sup_{X \in \mathscr{P}} \inner{C}{X} \leq \inf_{Y \in \mathscr{D}} \inner{D}{Y}.
\end{equation}
This is known as weak duality. On the other hand, if the so-called Slater's condition (see e.g.\cite{bookWatrous}) holds then the inequality reduces to an equality, and the SDP is said to satisfy strong duality.
The latter may fail to hold for some SDPs.
\medskip

We are now in a position to derive upper bounds on the unstabilized norms, and obtain SDP formulations of these bounds. We begin with the unstabilized norm $\| \cdot \|_{\infty}^{\mathrm{D}}$. Define the set
\begin{equation} \nonumber
    \pptd \coloneqq \{ \sigma \in \bhyx : \sigma \geq 0, \ptx{\sigma} \geq 0, \tr{\sigma} \leq 1 \},
\end{equation}
where $\ptx{\sigma}$ denotes the partial transpose of $\sigma$ over subsystem $A$.
Let $\Phi : \bhx \rightarrow \bhy$ be a Hermiticity-preserving linear map\footnote{Note that the difference of two quantum channels is given by such a map.} and define the following quantities:
\begin{equation} \label{eq:upperBound1ToInf}
    \begin{aligned}
        \mathcal{M}_{\infty}^{(-)}(\Phi) & \coloneqq \max_{\sigma \in \pptd} \{ - \tr{J(\Phi) \sigma} \}, \\
        \mathcal{M}_{\infty}^{(+)}(\Phi) & \coloneqq \max_{\sigma \in \pptd} \{ \tr{J(\Phi) \sigma} \},
    \end{aligned}
\end{equation}
where $J(\Phi) \in \bhyx$ is the Choi operator of $\Phi$. We have the following Proposition.
\begin{prop} \label{prop:ineq1ToInfMInf}
    Let $\Phi : \bhx \rightarrow \bhy$ be a Hermiticity-preserving linear map. Then
    \begin{equation} \label{eq:prop:ineq1ToInfMInf}
        \| \Phi \|_{\infty}^{\mathrm{D}} \leq \max \{ \mathcal{M}_{\infty}^{(-)}(\Phi), \mathcal{M}_{\infty}^{(+)}(\Phi) \}.
    \end{equation}
\end{prop}
\begin{proof}
    Using~\eqref{eq:Choi2}, we can write
    \begin{equation} \nonumber
        \| \Phi \|_{\infty}^{\mathrm{D}} = \max \{ ||\trx\left (J(\Phi) \left[\iy \otimes \rho \right] \right)||_{\infty} : \rho \in \dhx \}.
    \end{equation}
    Now, for any Hermitian operator $H$, one can write a spectral decomposition
    \begin{equation}
        H = \sum_i p_i \proj{\psi_i} - \sum_i q_i \proj{\phi_i},
    \end{equation}
    where $p_i \geq 0$ for all $i$ and $q_i > 0$ for all $i$. It can be seen that
    \begin{equation}
        \max_i p_i = \lambda_+(H) = \max_{Z \geq 0} \{ \tr{Z H} : \, \tr{Z} \leq 1 \},
    \end{equation}
    \begin{equation}
        \max_i q_i = - \lambda_-(H) = \max_{Z \geq 0} \{ -\tr{Z H} : \, \tr{Z} \leq 1 \},
    \end{equation}
    so that $||H||_{\infty} = \max \{ \lambda_+(H), - \lambda_-(H) \}$.
    Since $\Phi$ is Hermiticity-preserving, $J(\Phi)$ is Hermitian, and we can therefore write
    \begin{equation} \label{eq:opt1toInfD}
        \| \Phi \|_{\infty}^{\mathrm{D}} = \max \left\{ \max_{X \in \mathcal{Q}_1} \{ \Tr\left(J(\Phi) X \right) \}, \max_{X \in \mathcal{Q}_1} \{ - \Tr\left(J(\Phi) X \right) \} \right\}
    \end{equation}
    where we have defined the set
    \begin{equation} \label{eq:Q1}
        \mathcal{Q}_1 \coloneqq \{ Z \otimes \rho : \rho \in \dhx, Z \in \bhy, Z \geq 0, \tr{Z} \leq 1 \}.
    \end{equation}
    The fact that $\mathcal{Q}_1 \subseteq \pptd$ gives the desired result.
\end{proof}
\noindent Interestingly, the set $\pptd$ is convex, unlike the set $\mathcal{Q}_1$. This observation allows us to construct the following SDPs for the two quantities $\mathcal{M}_{\infty}^{(\pm)}(\Phi)$.
\begin{prop}
\label{prop:SDP-ub-infty}
    Let $\Phi : \bhx \rightarrow \bhy$ be a Hermiticity-preserving linear map. The optimization problems $\mathcal{M}_{\infty}^{(\pm)}(\Phi)$ (given in~\eqref{eq:upperBound1ToInf}) can be expressed as SDPs 
    with the following primal and dual problems:
    \begin{center}
    \begin{minipage}{.49\columnwidth}
    \begin{center}
        \underline{Primal problem}
        \begin{equation*}
            \begin{aligned}
                \text{max} \quad & \inner{\pm J(\Phi)}{\sigma}, \\
                \text{s.t.} \quad & \sigma \geq 0, \\
                & \Tr(\sigma) \leq 1, \\
                & \ptx{\sigma} \geq 0, \\
                & \sigma \in \bhyx.
            \end{aligned}
        \end{equation*}
    \end{center}
    \end{minipage}
    \begin{minipage}{.49\columnwidth}
    \begin{center}
        \underline{Dual problem}
        \begin{equation*}
            \begin{aligned}
                \text{min} \quad & \lambda, \\
                \text{s. t.} \quad & \lambda \in \mathbb{R}_+, \\
                & \omega \geq 0, \\
                & \lambda \, \iyx - \ptx{\omega} \geq \pm J(\Phi), \\
                & \omega \in \bhyx.
            \end{aligned}
        \end{equation*}
    \end{center}
    \end{minipage}
    \end{center}
\end{prop}
\begin{proof}
    The primal problems are obviously just rewritings of~\eqref{eq:upperBound1ToInf}.
    To see that the corresponding dual problems are as stated above, note that the primal and dual problems can be specified by the triple $(\Psi,C,D)$, where $\Psi : \bhyx \rightarrow \mathcal{L}(\mathbb{C} \oplus (\mathcal{H}_B \otimes \mathcal{H}_A))$ is a Hermiticity-preserving mapping defined through the following equation for any $\sigma \in  \bhyx$:
    \begin{equation}
        \Psi(\sigma) = \begin{pmatrix}
            \Tr(\sigma) & 0 \\
            0 & - \ptx{\sigma}
        \end{pmatrix},
    \end{equation}
    while $C \in \bhyx$ and $D \in \mathcal{L}(\mathbb{C} \oplus (\mathcal{H}_B \otimes \mathcal{H}_A))$ are Hermitian matrices defined as
    \begin{equation}
        C = \pm J(\Phi) \quad \text{ and } \quad D = \begin{pmatrix}
            1 & 0 \\
            0 & 0
        \end{pmatrix}.
    \end{equation}
    The adjoint of the map $\Psi$ is 
    defined through the relation
    \begin{equation}
        \Psi^* \begin{pmatrix}
            \lambda & \cdot \\
            \cdot & \omega
        \end{pmatrix} = \lambda \iyx - \ptx{\omega}.
    \end{equation}
    After a simplification of the primal and dual problems associated with $(\Psi,C,D)$, one obtains the primal and dual problems stated in Proposition~\ref{prop:SDP-ub-infty}.
\end{proof}

We now turn to the unstabilized norm $\| \cdot \|_{1}^{\mathrm{D}}$. Define the set
\begin{equation} \nonumber
    \begin{aligned}
        \ppto \coloneqq \{ & W \in \bhyx : W \geq 0, \ptx{W} \geq 0, \\
        &\quad W \leq \iy \otimes \rho \text{ for some } \rho \in \dhx \}.
    \end{aligned}
\end{equation}
Let $\Phi : \bhx \rightarrow \bhy$ be a linear map and define the two quantities
\begin{equation} \label{eq:upperBound1To1}
    \begin{aligned}
        \mathcal{M}_{1}^{(+)}(\Phi) & \coloneqq 2 \max_{W \in \ppto} \{ \tr{J(\Phi) W} \}\\
        \mathcal{M}_{1}^{(-)}(\Phi) & \coloneqq 2 \max_{W \in \ppto} \{ -\tr{J(\Phi) W} \}.
    \end{aligned}
\end{equation}
We have the following Proposition.
\begin{prop} \label{prop:ineq1To1M1}
    Let $\Phi : \bhx \rightarrow \bhy$ be the difference between two quantum channels. Then
    \begin{equation}
        \| \Phi \|_{1}^{\mathrm{D}} \leq \min \{\mathcal{M}_{1}^{(-)}(\Phi), \mathcal{M}_{1}^{(+)}(\Phi) \}.
    \end{equation}
\end{prop}
\begin{proof}
    For any traceless, Hermitian operator $H$,
    \begin{equation} \label{eq:traceNormHermTraceless}
        ||H||_1 = 2 \, \max_{Z \geq 0} \{ \tr{Z H} : \, Z \leq \mathds{1} \}.
    \end{equation}
    Since $\Phi$ is the difference of two quantum channels, we can write
    \begin{equation} \label{eq:opt1to1D}
        \begin{aligned}
            \| \Phi \|_{1}^{\mathrm{D}} & = 2 \max \{ \tr{\Phi(\rho) Z} : \rho \in \dhx, 0 \leq Z \leq \iy \} \\
            & = 2 \max_{X \in \mathcal{Q}_2} \{ \tr{J(\Phi) X} \},
        \end{aligned}
    \end{equation}
    where we used~\eqref{eq:Choi2} and defined the set
    \begin{equation} \label{eq:Q2}
        \mathcal{Q}_2 \coloneqq \{ Z \otimes \rho : \rho \in \dhx, Z \in \bhy, Z \geq 0, Z \leq \iy \}.
    \end{equation}
    The fact that $\mathcal{Q}_2 \subseteq \ppto$ implies that $\| \Phi \|_{1}^{\mathrm{D}} \leq \mathcal{M}_{1}^{(+)}(\Phi)$.

    {Similarly, for any traceless, Hermitian operator $H$, one can replace~\eqref{eq:traceNormHermTraceless} with
    \begin{equation} \label{eq:traceNormHermTraceless2}
        ||H||_1 = 2 \, \max_{Z \geq 0} \{ - \tr{Z H} : \, Z \leq \mathds{1} \},
    \end{equation}
    and one ends up with $\| \Phi \|_{1}^{\mathrm{D}} \leq \mathcal{M}_{1}^{(-)}(\Phi)$ with the same reasoning as above. Combining the two results gives the Proposition.
    }
\end{proof}
\noindent The fact that the set $\ppto$ is convex allows us to construct the following SDPs for the two quantities $\mathcal{M}_{1}^{(\pm)}(\Phi)$.
\begin{prop} \label{prop:primalDualUpperBound1to1}
    Let $\Phi : \bhx \rightarrow \bhy$ be a Hermiticity-preserving linear map. The optimization problems $\mathcal{M}_{1}^{(\pm)}(\Phi)$ (given in~\eqref{eq:upperBound1To1}) can be expressed as SDPs associated with the following primal and dual problems:
    \begin{center}
    \begin{minipage}{.49\columnwidth}
    \begin{center}
        \underline{Primal problem}
        \begin{equation*}
            \begin{aligned}
                \text{max} \quad & 2 \inner{\pm J(\Phi)}{W}, \\
                \text{s. t.} \quad & W \geq 0, \\
                & W \leq \iy \otimes \rho, \\
                & \ptx{W} \geq 0, \\
                & W \in \bhyx, \\
                & \rho \in \dhx.
            \end{aligned}
        \end{equation*}
    \end{center}
    \end{minipage}
    \begin{minipage}{.49\columnwidth}
    \begin{center}
        \underline{Dual problem}
        \begin{equation*}
            \begin{aligned}
                \text{min} \quad & ||\try(Y_2)||_{\infty}, \\
                \text{s. t.} \quad & Y_1, Y_2 \geq 0, \\
                & Y_2 - \ptx{Y_1} \geq \pm 2 J(\Phi), \\
                & Y_1, Y_2 \in \bhyx. \\
                & ~\\
                & ~
            \end{aligned}
        \end{equation*}
    \end{center}
    \end{minipage}
    \end{center}
\end{prop}
\begin{proof}
    The primal problems are obviously just rewritings of~\eqref{eq:upperBound1To1}. To see that the corresponding dual problems are as stated above, note that the primal and dual problems can be specified by the triple $(\Psi,C,D)$, where $\Psi : \mathcal{L}( (\mathcal{H}_B \otimes \mathcal{H}_A) \oplus \mathcal{H}_A) \rightarrow \mathcal{L}((\mathcal{H}_B \otimes \mathcal{H}_A) \oplus (\mathcal{H}_B \otimes \mathcal{H}_A) \oplus \mathbb{R})$ is a Hermiticity-preserving mapping defined as
    \begin{equation}
        \Psi \begin{pmatrix}
            W & 0 \\
            0 & \rho
        \end{pmatrix} = \begin{pmatrix}
            - \ptx{W} & 0 & 0 \\
            0 & W - \iy \otimes \rho & 0 \\
            0 & 0 & \tr{\rho}
        \end{pmatrix},
    \end{equation}
    while $C \in \mathcal{L}( (\mathcal{H}_B \otimes \mathcal{H}_A) \oplus \mathcal{H}_A)$ and $D \in \mathcal{L}((\mathcal{H}_B \otimes \mathcal{H}_A) \oplus (\mathcal{H}_B \otimes \mathcal{H}_A) \oplus \mathbb{R})$ are Hermitian matrices defined as
    \begin{equation}
        C = 2 \begin{pmatrix}
            \pm J(\Phi) & 0 \\
            0 & 0
        \end{pmatrix} \quad \text{ and } \quad D = \begin{pmatrix}
            0 & 0 & 0 \\
            0 & 0 & 0 \\
            0 & 0 & 1
        \end{pmatrix}.
    \end{equation}
    The adjoint of the map $\Psi$ is given by
    \begin{equation}
        \Psi^* \begin{pmatrix}
            Y_1 & \cdot & \cdot \\
            \cdot & Y_2 & \cdot \\
            \cdot & \cdot & \lambda
        \end{pmatrix} = \begin{pmatrix}
            Y_2 - \ptx{Y_1} & 0 \\
            0 & \lambda \ix - \try(Y_2)
        \end{pmatrix}.
    \end{equation}
    From this, the dual problem can be written as
    \begin{center}
        \underline{Dual problem}
        \begin{equation*}
            \begin{aligned}
                \text{min} \quad & \lambda, \\
                \text{s. t.} \quad & \lambda \geq 0, \\
                & Y_1, Y_2 \geq 0, \\
                & \lambda \ix - \try(Y_2) \geq 0, \\
                & Y_2 - \ptx{Y_1} \geq \pm 2 J(\Phi), \\
                & Y_1, Y_2 \in \bhyx.
            \end{aligned}
        \end{equation*}
    \end{center}
    Now, since $\lambda \geq 0$ and $Y_2 \geq 0$,
    \begin{equation}
        \lambda \ix \geq \try(Y_2) \quad \Leftrightarrow \quad \lambda \geq ||\try(Y_2)||_{\infty}.
    \end{equation}
   We can therefore simplify the above dual problem to the one given in the Proposition. 
\end{proof}

\begin{lem} \label{lem:M1leqDia}
    Let $\Phi : \bhx \rightarrow \bhy$ be the difference between two quantum channels. Then
    \begin{equation} \label{eq:lem:M1leqDia}
        \begin{aligned}
            2 \max \{ \mathcal{M}_{\infty}^{(-)}(\Phi), \mathcal{M}_{\infty}^{(+)}(\Phi) \} & \leq \max \{ \mathcal{M}_{1}^{(-)}(\Phi), \mathcal{M}_{1}^{(+)}(\Phi) \} \\
            & \leq ||\Phi||_{\dia}.
        \end{aligned}
    \end{equation}
\end{lem}
\begin{proof}
    We begin by proving the second inequality in~\eqref{eq:lem:M1leqDia}. Watrous showed that if $\Phi$ is the difference between two quantum channels, then (see Section~4 in~\cite{Watrous2009})
    \begin{equation}
        \| \Phi \|_{\dia} = 2 \max_{X \in \mathcal{Q}_3} \{ \tr{J(\Phi) X} \},
    \end{equation}
    where
    \begin{equation}
        \begin{aligned}
            \mathcal{Q}_3 \coloneqq \{ W : & W \in \bhyx, W \geq 0, \\
            & W \leq \iy \otimes \rho \text{ for some } \rho \in \dhx \}.
        \end{aligned}
    \end{equation}
    The fact that $\ppto \subseteq \mathcal{Q}_3$ then implies $\mathcal{M}_{1}^{(+)}(\Phi) \leq ||\Phi||_{\dia}$. Following Watrous' proof, it can also be shown that
    \begin{equation}
        \| \Phi \|_{\dia} = 2 \max_{X \in \mathcal{Q}_3} \{ - \tr{J(\Phi) X} \},
    \end{equation}
    so that we also have $\mathcal{M}_{1}^{(-)}(\Phi) \leq ||\Phi||_{\dia}$. Combining the two previous inequalities gives the second inequality in~\eqref{eq:lem:M1leqDia}.
    
    To prove the first inequality in~\eqref{eq:lem:M1leqDia}, first note that by the so-called {\em{reduction criterion}}~\cite{Horodecki1999}, if a quantum state $\sigma_{BA} \in \bhyx$ is not distillable, it satisfies $\sigma_{BA} \leq \iy \otimes \sigma_A$, with $\sigma_A = \try(\sigma_{BA})$. Furthermore, it is known that PPT states are not distillable. By combining these two facts, we can infer that $\pptd \subseteq \ppto$. We therefore have that $2 \mathcal{M}_{\infty}^{(-)}(\Phi) \leq \mathcal{M}_{1}^{(-)}(\Phi)$ and $2 \mathcal{M}_{\infty}^{(+)}(\Phi) \leq \mathcal{M}_{1}^{(+)}(\Phi)$. Combining these two inequalities gives the desired result.
\end{proof}

{
\begin{rem} \label{rem:4}
From the SDP given in Section~3 of~\cite{Watrous2013} for the diamond norm, one can compute $\| \Phi^{\mathrm{c}} - \Lambda \circ \Phi \|_{\dia}$ for any channel $\Lambda$. The SDPs provided in Propositions~\ref{prop:SDP-ub-infty} and~\ref{prop:primalDualUpperBound1to1} also allow one to compute the four quantities $\mathcal{M}_{1}^{(-)}(\Phi)$, $\mathcal{M}_{1}^{(+)}(\Phi)$, $\mathcal{M}_{\infty}^{(-)}(\Phi)$ and $\mathcal{M}_{\infty}^{(+)}(\Phi)$. According to Propositions~\ref{prop:ineq1ToInfMInf} and~\ref{prop:ineq1To1M1}, one can then take $\eps_1 = \max \{ \mathcal{M}_{1}^{(-)}(\Phi), \mathcal{M}_{1}^{(+)}(\Phi) \}$, $\eps_{\dia } = \| \Phi^{\mathrm{c}} - \Lambda \circ \Phi \|_{\dia}$ and $\nu = \max \{ \mathcal{M}_{\infty}^{(-)}(\Phi), \mathcal{M}_{\infty}^{(+)}(\Phi) \}$ in {Propositions~\ref{prop:upperQuantCapa2} and~\ref{prop:upperPrivCapa2}}, since we trivially have that $\max \{ \mathcal{M}_{1}^{(-)}(\Phi), \mathcal{M}_{1}^{(+)}(\Phi) \} \geq \min \{ \mathcal{M}_{1}^{(-)}(\Phi), \mathcal{M}_{1}^{(+)}(\Phi) \}$.
Now, Lemma~\ref{lem:M1leqDia} implies that $\beta = 2\nu/\eps_1 \leq 1$ and $\eps_1 \leq \eps_\dia$. This in turn guarantees that the bounds given in~\eqref{eq:prop:upperQuantCapa2} {and~\eqref{eq:prop:upperPrivCapa2}} (for these choices of $\eps_1$, $\eps_{\dia }$ and $\nu$) are at least as good as the bound given in points~$(i)$ {and~$(iii)$} of Theorem~7 in~\cite{Sutter2017}.
\end{rem}
}


\subsection{An application: {upper bound on the quantum capacity of the qubit depolarizing channel}\label{sec:depol}}

The qubit depolarizing channel $\mathcal{E}_p : A \rightarrow B$ with $\mathcal{H}_A\simeq \mathcal{H}_B \simeq \mathbb{C}^{2}$ and $p \in [0,1]$ is defined as
\begin{equation}
    \mathcal{E}_p(\rho) \coloneqq (1-p) \rho + \frac{p}{3} (X \rho X + Z \rho Z + Y \rho Y),
\end{equation}
where $X \coloneqq \begin{pmatrix}
        0 & 1 \\
        1 & 0
    \end{pmatrix}$,
    $Y \coloneqq \begin{pmatrix}
        0 & -i \\
        i & 0
    \end{pmatrix}$, and
    $Z \coloneqq \begin{pmatrix}
        1 & 0 \\
        0 & -1
    \end{pmatrix}$
are the Pauli matrices.
Its channel coherent information can be computed to be~\cite{Bennett1996,DiVincenzo1998},
\begin{equation} \label{eq:cohInfoDepol}
    Q^{(1)}(\mathcal{E}_p) = 1 + (1-p) \log (1-p) + p \log (p/3).
\end{equation}
The following upper bound for its quantum capacity was obtained in~\cite{Smith2008} for $p \in [0,1/4]$:
\begin{equation}
    Q(\mathcal{E}_p) \leq \mathrm{conv} \left\{ 1 - h(p), c(p), 1 - 4p \right\},
\end{equation}
where $\theta(p) \coloneqq h ([1 + \gamma(p)]/2) - h (\gamma(p)/2)$, $\gamma(p) \coloneqq 4(\sqrt{1\!-\!p}\!-\!1\!+\!p)$, and $\mathrm{conv} \{ f_1(p), \cdots, f_n(p)\}$ denotes the maximal convex function that is less than or equal to all $f_i(p)$, for $i=1,\cdots,n$.
Exploiting the $\eps$--degradability of the depolarizing channel, the authors of~\cite{Sutter2017} improved the above upper bound to the following (for $p \in [0,1/4]$):
\begin{equation}
    Q(\mathcal{E}_p) \leq \mathrm{conv} \left\{ S(\Phi,\Lambda), 1 - h(p), \theta(p), 1 - 4p \right\},
\end{equation}
where $\Lambda$ denotes any degrading channel and $S(\Phi,\Lambda)$ is defined through~\eqref{eq:maxCondEnt}.
The above bound was then simplified with the help of Proposition~5 of~\cite{Sutter2017} to
\begin{align}
        Q(\mathcal{E}_p) \leq \, \mathrm{conv} \bigl\{ & Q^{(1)}(\mathcal{E}_p) + \frac{\eps^{(p)}_{\dia}}{2} \log (d_E-1) + h (\eps^{(p)}_{\dia}/2) \nonumber \\
        & \quad \quad \quad \quad \quad + \eps^{(p)}_{\dia} \log(d_E) + g(\eps^{(p)}_{\dia}/2), \nonumber \\
        & \quad \quad \quad \quad \quad 1 - h(p), \theta(p), 1 - 4p \bigr\}, \label{eq:boundQSutter}
\end{align}
with $d_E = \mathrm{dim} \mathcal{H}_E$ and where
\begin{equation} \label{eq:opteps}
    \eps^{(p)}_{\dia} \coloneqq \inf_{\Lambda} \left\{ || \mathcal{E}_p^{\mathrm{c}} - \Lambda \circ \mathcal{E}_p ||_{\dia} : \Lambda : B \to E\, \text{is CPTP} \right\}.
\end{equation}

It turns out that the bound of~\eqref{eq:boundQSutter} can be further improved by exploiting the $(\eps,\nu)$-degradability of the depolarizing channel. Define $\Lambda_p$ to be the degrading channel achieving the infimum over $\Lambda$ in~\eqref{eq:opteps}, i.e., the degrading channel that satisfies
\begin{equation}\label{eq:epdia}
    \eps^{(p)}_{\dia} = || \mathcal{E}_p^{\mathrm{c}} - \Lambda_p \circ \mathcal{E}_p ||_{\dia}.
\end{equation}
Next, as explained in Remark~\ref{rem:4}, compute the upper bounds on $|| \mathcal{E}_p^{\mathrm{c}} - \Lambda_p \circ \mathcal{E}_p ||_1^{\mathrm{D}}$ and $|| \mathcal{E}_p^{\mathrm{c}} - \Lambda_p \circ \mathcal{E}_p ||_{\infty}^{\mathrm{D}}$, which are respectively given by
\begin{equation}\label{eq:ep1nu}
    \begin{aligned}
        \eps^{(p)}_1 & \coloneqq \max \{ \mathcal{M}_{1}^{(-)}(\Phi_p), \mathcal{M}_{1}^{(+)}(\Phi_p) \}, \\
        \nu^{(p)} & \coloneqq \max \{ \mathcal{M}_{\infty}^{(-)}(\Phi_p), \mathcal{M}_{\infty}^{(+)}(\Phi_p) \},
    \end{aligned}
\end{equation}
where $\Phi_p \coloneqq \mathcal{E}_p^{\mathrm{c}} - \Lambda_p \circ \mathcal{E}_p$.
By combining Proposition~\ref{prop:upperQuantCapa2} with the other aforementioned bounds on the quantum capacity of the qubit depolarizing channel, we obtain the new upper bound
\begin{align}
        Q(\mathcal{E}_p) \leq \, \mathrm{conv} \bigl\{ & Q^{(1)}(\mathcal{E}_p) \! + \!  \frac{\eps^{(p)}_{1}}{2}\log ( \beta^{(p)} d_E \! - \! 1 ) \! + \! h (\eps^{(p)}_{1}/2) \nonumber \\
        & \quad \quad \quad \quad \quad + \eps^{(p)}_{\dia} \log(d_E) + g(\eps^{(p)}_{\dia}/2), \nonumber \\
        & \quad \quad \quad \quad \quad 1 - h(p), \theta(p), 1 - 4p \bigr\}, \label{eq:boundQnu}
\end{align}
where $\beta^{(p)} \coloneqq 2\nu^{p)}/\eps^{(p)}_1$. Note that, by the hypothesis of Proposition~\ref{prop:upperQuantCapa2}, the bound in~\eqref{eq:boundQnu} is only valid if $\eps^{(p)}_1 \leq 2 \nu^{(p)} d_E/(\nu^{(p)} d_E + 3)$, which can be checked after computing $\eps^{(p)}_1$ and $\nu^{(p)}$. This inequality is satisfied at least for the low-noise range $p \in [0, 0.025]$, for which we compute the upper bound on the quantum capacity
{(see Fig.~\ref{fig:norms_const} in Appendix~\ref{sec:appB}).}
Note that, for this parameter range, \eqref{eq:boundQSutter} was known to be the tightest bound on the quantum capacity ({\em{cf.}}~\cite{Leditzky2018,Fanizza2020}). Figs.~\ref{fig:depol1} and~\ref{fig:depol2} compare our new upper bound given in~\eqref{eq:boundQnu} with the previous upper bound~\eqref{eq:boundQSutter}.

\section{Open Problems\label{sec:discussions}}

We conclude the paper by stating a conjecture and some open problems.
The continuity bound given in Theorem~\ref{theo:ContinuityOpNorm} is not tight when $1/\beta$ is not an integer. We would like to conjecture the following uniform continuity bound holds and is tight.
\begin{conj} \label{conj:ContinuityOpNorm2:ContinuityOpNorm2}
Consider a Hilbert space ${\mathcal{H}}$ of dimension $d=|{\mathcal{H}}|<\infty$. Let $\rho$ and $\sigma$ be two states on $\mathcal{H}$ and let $\eps \coloneqq \mathrm{T}(\rho,\sigma)$ and $\nu \coloneqq ||\rho-\sigma||_{\infty}$.
Then the following inequality holds and is tight:
\begin{equation} \label{eq:conj:ContinuityOpNorm2}
	\left| S(\rho) - S(\sigma) \right| \leq f_{d}(\eps,\nu),
\end{equation}
where the function $f_d$ is defined through~\eqref{eq:RHSBound}.
\end{conj}
\begin{figure}[!t]
\centering
\subfloat[]{\includegraphics[width=\columnwidth]{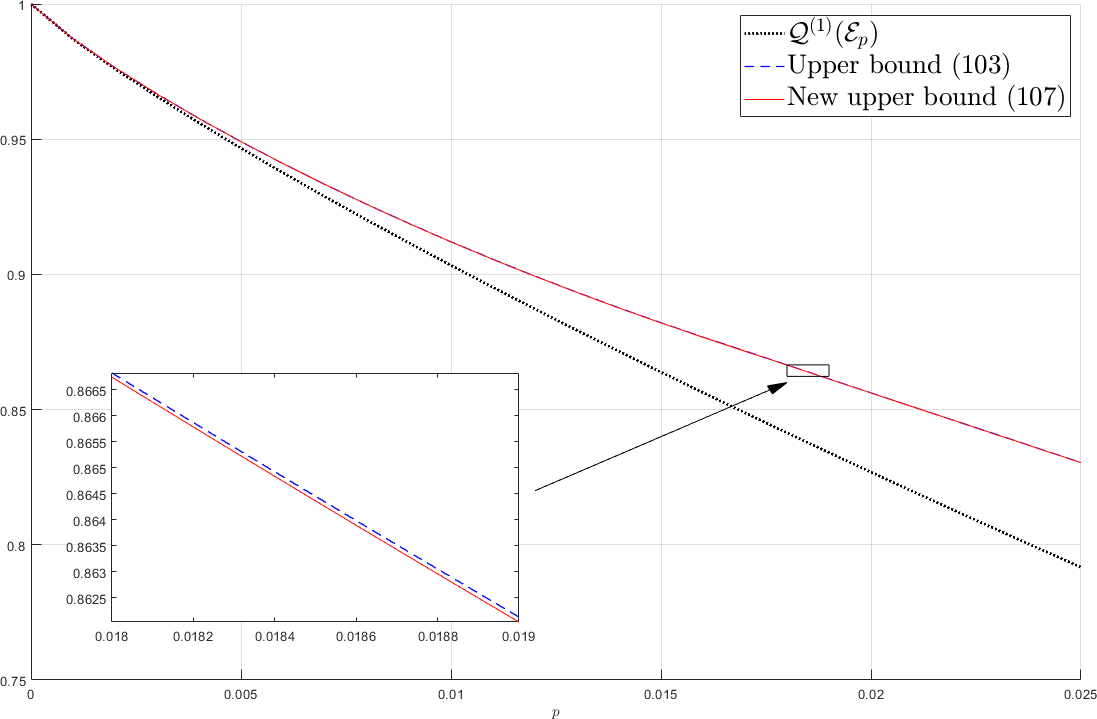}
\label{fig:depol1}}
\hfil
\subfloat[]{\includegraphics[width=\columnwidth]{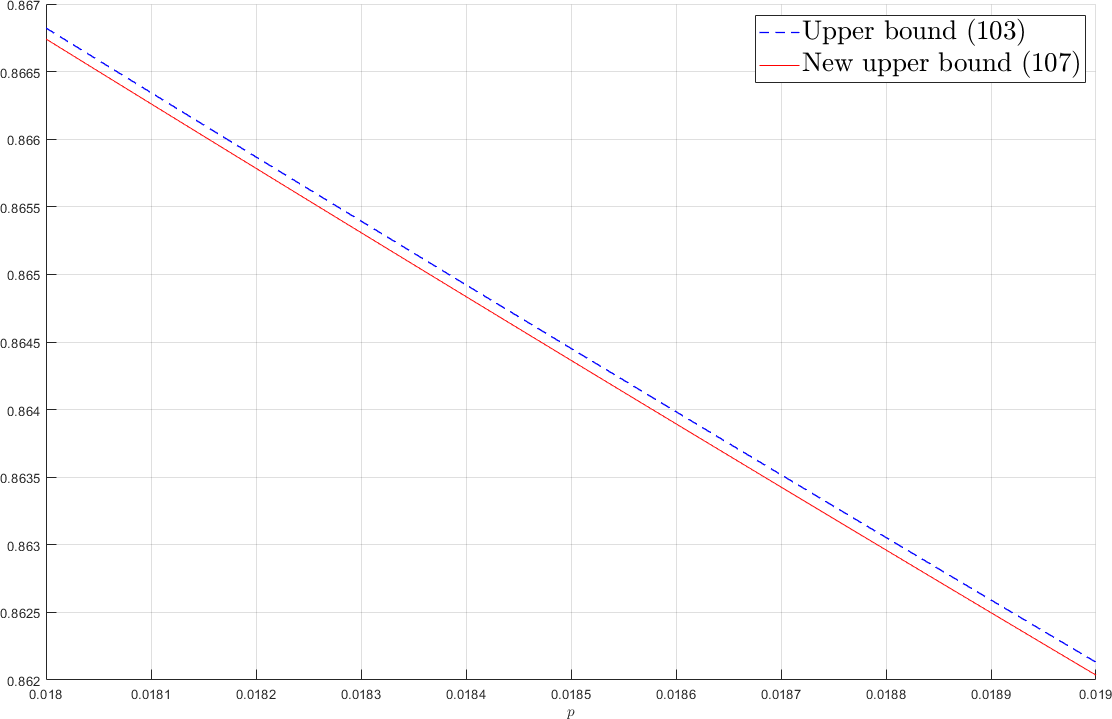}
\label{fig:depol2}}
\caption{(a) depicts upper and lower bounds for the quantum capacity $Q(\mathcal{E}_p)$ of a qubit depolarizing channel $\mathcal{E}_p$ in the low-noise regime $p \in [0, 0.025]$ for which~\eqref{eq:boundQSutter} (dashed blue curve) was the previously known tightest upper bound. The channel coherent information given in~\eqref{eq:cohInfoDepol} denotes a lower bound on $Q(\mathcal{E}_p)$ (dotted black curve). The solid red line depicts our new upper bound given by~\eqref{eq:boundQnu}. (b): a magnification of the box represented on Fig.~\ref{fig:depol1}.}
\end{figure}
\noindent A possible way to prove Conjecture~\ref{conj:ContinuityOpNorm2:ContinuityOpNorm2}, would be to proceed in a manner analogous to the proof of Theorem~\ref{theo:ContinuityOpNorm}. Recall that the proof of the latter is based on the classical result of Lemma~\ref{lem:Sason}. Similarly, one could attempt to prove Conjecture~\ref{conj:ContinuityOpNorm2:ContinuityOpNorm2} by exploiting the strengthened classical result of Theorem~\ref{theo:main}. However, one crucial step of the proof would be to establish that there exists some $\eps^* \in [0,1]$ for which the function $f_{d}(\eps,\nu)$ is monotonically increasing in $\eps \in [0,\eps^*]$, for any fixed $\nu \in [0,\eps^*]$.
This would be analogous to the steps corresponding 
to~\eqref{eq:prooftheo:ContinuityOpNorm2} and~\eqref{eq:prooftheo:ContinuityOpNorm3} in the proof of Theorem~\ref{theo:ContinuityOpNorm}. It turns out, however, that it is not possible to find such a value $\eps^*$. This can be better understood from
{Figs.~\ref{fig:RHSa} and~\ref{fig:RHSb} in Appendix~\ref{sec:appC}}
which depict 3D plots of $f_{d}(\eps,\nu)$ for $d=10$, and
{Fig.~\ref{fig:RHS2} in Appendix~\ref{sec:appC},}
which depicts 2D plots of $f_{d}(\eps,\nu)$ for $d=10$, for fixed values of $\nu$.
Note that the continuity bound of~\eqref{eq:conj:ContinuityOpNorm2}, if it holds, can easily be shown to be tight for all allowed values of $\eps$ and $\nu$. To see this, when $\mu=0$ in the definition~\eqref{eq:RHSBound} of $f_{d}(\eps,\nu)$, consider a pair of commuting states, $\rho$ and $\sigma$ whose vectors of eigenvalues are given by the probability distributions of~\eqref{eq:probaTightMe1} and~\eqref{eq:probaTightMe2a}, respectively; when $\mu \neq 0$, consider instead states $\rho$ and $\sigma$ with vectors of eigenvalues given by the probability distributions of~\eqref{eq:probaTightMe1} and~\eqref{eq:probaTightMe2b}, respectively.

To our knowledge, for a linear map, $\Phi$, which is the difference of two quantum channels, the unstablized norm $||\Phi||_{\infty}^{\mathrm{D}}$ (equivalently, the optimization problem in~\eqref{eq:opt1toInfD}) cannot be expressed as an SDP. This is because the set $\mathcal{Q}_1$ defined in~\eqref{eq:Q1} is not convex. Similarly, the unstablized norm $||\Phi||_1^{\mathrm{D}}$ (equivalently, the optimization problem in~\eqref{eq:opt1to1D}) cannot be expressed as an SDP either, as the set $\mathcal{Q}_2$ defined in~\eqref{eq:Q2} is not convex. However, if we were able to compute the two norms, $||\Phi||_1^{\mathrm{D}}$ and $||\Phi||_{\infty}^{\mathrm{D}}$, then we would be able to obtain better upper bounds on the quantum capacities of channels through Proposition~\ref{prop:upperQuantCapa2}. Therefore, it would of interest to design methods to compute these two norms. One possibility might be to show that they can be expressed via a so-called \emph{second-order cone program}~\cite{Alizadeh2003}.

The bounds for approximate degradable channels provided in~\cite{Sutter2017} were exploited in~\cite{Leditzky2018b} to determine the quantum and the private capacity of low-noise quantum channels to leading orders in the channel's distance to the perfect channel (we denote it by $\eta$). To this end, the authors of~\cite{Leditzky2018b} derived analytical upper bounds on the quantity $ \eps_\Phi$ in~\eqref{eq:optepsGen} to leading orders in $\eta$. An interesting question is whether their techniques can be applied to furthermore provide upper bounds on the quantity $\nu_\Phi$ in~\eqref{eq:optnuGen} to leading orders in $\eta$, or even the quantity
\begin{equation} \nonumber
    \beta_\Phi \coloneqq 2 \inf_{\Lambda} \left\{ \frac{|| \Phi^{\mathrm{c}} - \Lambda \circ \Phi ||_{\cb}}{|| \Phi^{\mathrm{c}} - \Lambda \circ \Phi ||_{\dia}} : \Lambda : \mathcal{L}(B) \rightarrow \mathcal{L}(E) \, \text{is CPTP} \right\}
\end{equation}
that would replace $\beta$ in~\eqref{eq:theo:upperQuantCapa}. This would in turn strengthen the upper bounds on the quantum- and private classical capacity derived in~\cite{Leditzky2018b}.

It would also be interesting to see whether ideas analogous to the ones developed in this paper could be used to get improved upper bounds on the classical capacity of certain quantum channels (see e.g.~\cite{Leditzky2018c,Leung2009,Shirokov2017}).


\appendices

\section{Supplemental lemmas} \label{sec:appA}

\noindent The following Lemma generalizes Lemma~4 of \cite{JabbourDatta}.
\begin{lem} \label{lem:majEps2}
    Consider a vector $\bos{v} \in \mathbb{R}^n$ whose elements are non-negative. For some $i,j \in \left\lbrace 1, \cdots, n \right\rbrace$ such that $v_i \geq v_j$, and some $s \in (0, v_j]$, define the vector $\bos{u} \in \mathbb{R}^n$ whose elements satisfy the following relations:
    \begin{equation}
    \left\lbrace \begin{aligned}
        u_i & = v_i + s, \\
		u_j & = v_j - s, \\
		u_k & = v_k, \quad \forall k \in \left\lbrace 1, \cdots, n \right\rbrace \setminus \left\lbrace i,j \right\rbrace.
    \end{aligned} \right.
    \end{equation}
	Then $\bos{v} \prec \bos{u}$.
\end{lem}
\begin{proof}
    Since $v_i \geq v_j$, the vector $(v_i,v_j)$ is majorized by the vector $(u_i,u_j)$. Lemma~\ref{lem:majEps2} then follows from Lemma~5 of \cite{JabbourDatta}.
\end{proof}

\noindent We use Lemma~\ref{lem:majEps2} to prove the following.
\begin{lem} \label{lem:majq}
    Consider a vector $\bos{v} \in \mathbb{R}^n$ whose elements are non-negative and are arranged in non-increasing order, and for some $i \in \left\lbrace 1, \cdots, n-1 \right\rbrace$, define the vector $\bos{u} \in \mathbb{R}^n$ whose elements satisfy the following relations:
    \begin{equation}
    \left\lbrace \begin{aligned}
        u_j & = v_j + \omega_j, \quad \forall j \in \left\lbrace 1, \cdots, i \right\rbrace, \\
		u_j & = v_j - \omega_j, \quad \forall j \in \left\lbrace i+1, \cdots, n \right\rbrace,
    \end{aligned} \right.
    \end{equation}
    where $\omega_j \geq 0$ for all $j \in \left\lbrace 1, \cdots, n \right\rbrace$ and $\sum_{j=1}^{i} \omega_j = \sum_{j=i+1}^{n} \omega_j$. Then $\bos{v} \prec \bos{u}$.
\end{lem}
\begin{proof}
    It is clear that the vector $\bos{u}$ can be obtained from vector $\bos{v}$ via a series of steps, where 
    in each step one removes some weight $s$ from an element in $\left\lbrace i+1, \cdots, n \right\rbrace$ of the initial vector, and adds that weight $s$ to an element in $\left\lbrace 1, \cdots, i \right\rbrace$. Since during the whole process, elements in $\left\lbrace 1, \cdots, i \right\rbrace$ always have values higher than elements in $\left\lbrace i+1, \cdots, n \right\rbrace$, Lemma~\ref{lem:majEps2} implies that in each step, the initial vector is majorized by the final vector. Since majorization is transitive, it then follows that $\bos{v} \prec \bos{u}$.
\end{proof}
\begin{lem} \label{lem:MeLeqSason}
    Let $M > 1$.
    Define the function $\tilde{g}_M : \mathcal{S} \rightarrow [0, \infty)$ as
    \begin{equation} \label{eq:RHSBoundSason}
        \tilde{g}_M(\eps,\nu) \coloneqq h(\eps) + \eps \log(\beta M-1),
    \end{equation}
    where the set $\mathcal{S}$ is defined as
    \begin{equation} \label{eq:setnuepsapp}
        \mathcal{S} \coloneqq \{ (\eps,\nu) \in [0, 1] \times [0, 1] : \nu \leq \eps \},
    \end{equation}
    and $\beta \coloneqq \nu/\eps$, and consider the function $f_M : \mathcal{S} \rightarrow [0, \infty)$ defined as
    \begin{equation} \label{eq:RHSBoundapp}
        f_d(\eps,\nu) \coloneqq h(\eps) + d_+ \, \nu \log \nu + \mu \log \mu + \eps \log d_- -\eps \log \eps,
    \end{equation}
    where $d_+ \in \mathbb{Z}$ and $\mu \in [0,\nu)$ are uniquely defined through the following relations:
    \begin{equation} \label{eq:remapp}
        \eps = d_+ \nu + \mu,
    \end{equation}
    and
    \begin{equation}
        d_- \coloneqq \begin{cases}
	   d-d_+, & \mathrm{if} \, \mu = 0, \\
	   d-d_+-1, & \mathrm{else}.
        \end{cases}
    \end{equation}
    Then $\tilde{g}_M(\eps,\nu) = f_M(\eps,\nu)$ if $\eps/\nu$ is an integer. Furthermore, if $\nu M \geq 2 \eps$ and $\eps/\nu$ is not an integer, then $\tilde{g}_M(\eps,\nu) > f_M(\eps,\nu)$.
\end{lem}
\begin{proof}
    When $\eps/\nu = 1/\beta$ is an integer, $d_+ = 1/\beta$ and $\mu=0$, so that
    \begin{equation}
            f_M(\eps,\nu) = \tilde{g}_M(\eps,\nu).
    \end{equation}
    When $1/\beta = \eps/\nu$ is not an integer, $\mu \neq 0$. In that case, some simplifications lead to
\begin{equation} \nonumber
    \tilde{g}_M(\eps,\nu) - f_M(\eps,\nu) = - \eps \log \left[ 1-\frac{\nu-\mu}{\nu M-\eps} \right] + \mu \log \left( \frac{\nu}{\mu} \right).
\end{equation}
Now, $\nu > \mu$, so that the second term of the above equation is strictly positive. Furthermore, $\nu M \geq 2 \eps$ implies $\nu M - \eps \geq 0$, which means that the first term of the above equation is non-negative. Consequently, $\tilde{g}_M(\eps,\nu) > f_M(\eps,\nu)$.
\end{proof}
\begin{lem} \label{lem:incBoundSason1}
    Let $M > 1$ and $\eps \in [0,1]$. Define the function $\textsl{g}_M : [\eps/M, 1] \rightarrow [0, \infty)$ as 
    \begin{equation} \label{eq:RHSBoundSason1}
        \textsl{g}_M(\nu) \coloneqq \eps \log\left((\nu/\eps) M-1\right) + h(\eps).
    \end{equation}
    Then $\textsl{g}_M$ is monotonically increasing in $\nu$ for all $\nu \in [\eps/M,1]$.
\end{lem}
%
\begin{proof}
    The claim simply follows from the monotonicity of the logarithm.
\end{proof}
\begin{lem} \label{lem:incBoundSason2}
    Let $M > 1$ and $\nu \in [0,1]$. Define the function $\textsl{f}_M : [0, \nu M/2] \rightarrow [0, \infty)$ as 
    \begin{equation} \label{eq:RHSBoundSason2}
        \textsl{f}_M(\eps) \coloneqq \eps \log((\nu/\eps) M-1) + h(\eps).
    \end{equation}
    Then $\textsl{f}_M$ is monotonically increasing in $\eps$ for all $\eps \in [0,\nu M / (\nu M + 3)]$.
\end{lem}
\begin{proof}
    We compute the derivative
    \begin{equation}
    \frac{\partial \textsl{f}_M}{\partial \eps} = \log \left( \frac{(1-\eps)(\nu M - \eps)}{\eps^2} \right) - \frac{1}{\ln 2} \frac{\nu M}{\nu M - \eps},
    \end{equation}
    where $\ln$ represents the natural logarithm.
    The above can be rewritten as
    \begin{equation}
        \frac{\partial \textsl{f}_M}{\partial \eps} = \log \left( \frac{(1-\eps)(\nu M - \eps)}{(2^{1/\ln 2}) \eps^2} \right) - \frac{1}{\ln 2} \frac{\eps}{\nu M - \eps}.
    \end{equation}
    Since $\eps \leq 1$, we have that $(2^{1/\ln 2}) \eps^2 \leq \eps(1+2\eps)$, so that, by monotonicity of the logarithm,
    \begin{equation} \label{eq:lemma5-1}
        \frac{\partial \textsl{f}_M}{\partial \eps} \geq \log \left( \frac{(1-\eps)(\nu M - \eps)}{\eps(1+2\eps)} \right) - \frac{1}{\ln 2} \frac{\eps}{\nu M - \eps}.
    \end{equation}
    Now, the well known lower bound for the logarithm function
    \begin{equation}
        \ln x \geq 1-\frac{1}{x}, \quad x \geq 0,
    \end{equation}
    can be expressed as
    \begin{equation}
        y \leq \ln \left( \frac{1}{1-y} \right) = \ln(2) \log \left(\frac{1}{1-y}\right), \quad y \leq 1.
    \end{equation}
    Since $\eps \leq \nu M/2$, we have that $\eps/(\nu M - \eps) \leq 1$, so that we can exploit the above inequality and write
    \begin{equation}
        \frac{\eps}{\nu M - \eps} \leq \ln (2) \log \left( \frac{\nu M - \eps}{\nu M - 2 \eps}\right).
    \end{equation}
    Combining this with~\eqref{eq:lemma5-1}, we obtain
    \begin{equation}
            \frac{\partial \textsl{f}_M}{\partial \eps} \geq \log \left( \frac{(1-\eps)(\nu M - \eps)}{\eps(1+2\eps)} \right) - \log \left(\frac{\nu M - \eps}{\nu M - 2 \eps}\right).
    \end{equation}
    This implies that the derivative of $\textsl{f}_M$ will be non-negative if the right hand side of the above inequality is non-negative. This will be the case if and only if $\eps \leq \nu M / (\nu M + 3)$, which holds by the hypothesis of the lemma.
\end{proof}

\section{Validity of the upper bound on the quantum capacity of the depolarizing channel} \label{sec:appB}

{In Fig.~\ref{fig:norms_both}, we present some plots illustrating the fact that the condition $\eps^{(p)}_1 \leq 2 \nu^{(p)} d_E/(\nu^{(p)} d_E + 3)$ is verified (see end of Section~\ref{sec:depol}), confirming the validity of the bound in~\eqref{eq:boundQnu}.}

\begin{figure}[!h]
\centering
\subfloat[]{\includegraphics[width=.94\columnwidth]{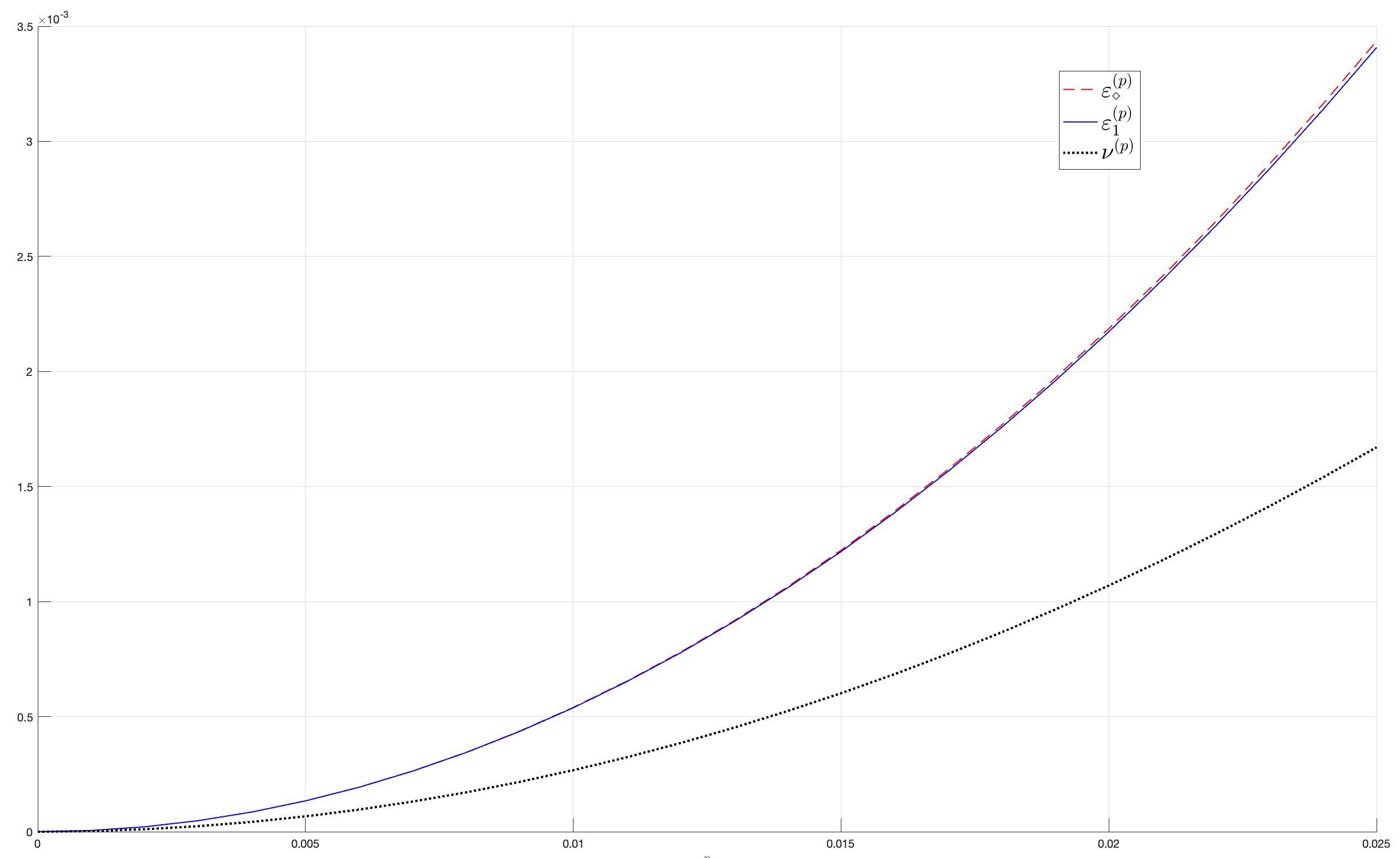}
\label{fig:norms}}
\hfil
\subfloat[]{\includegraphics[width=.94\columnwidth]{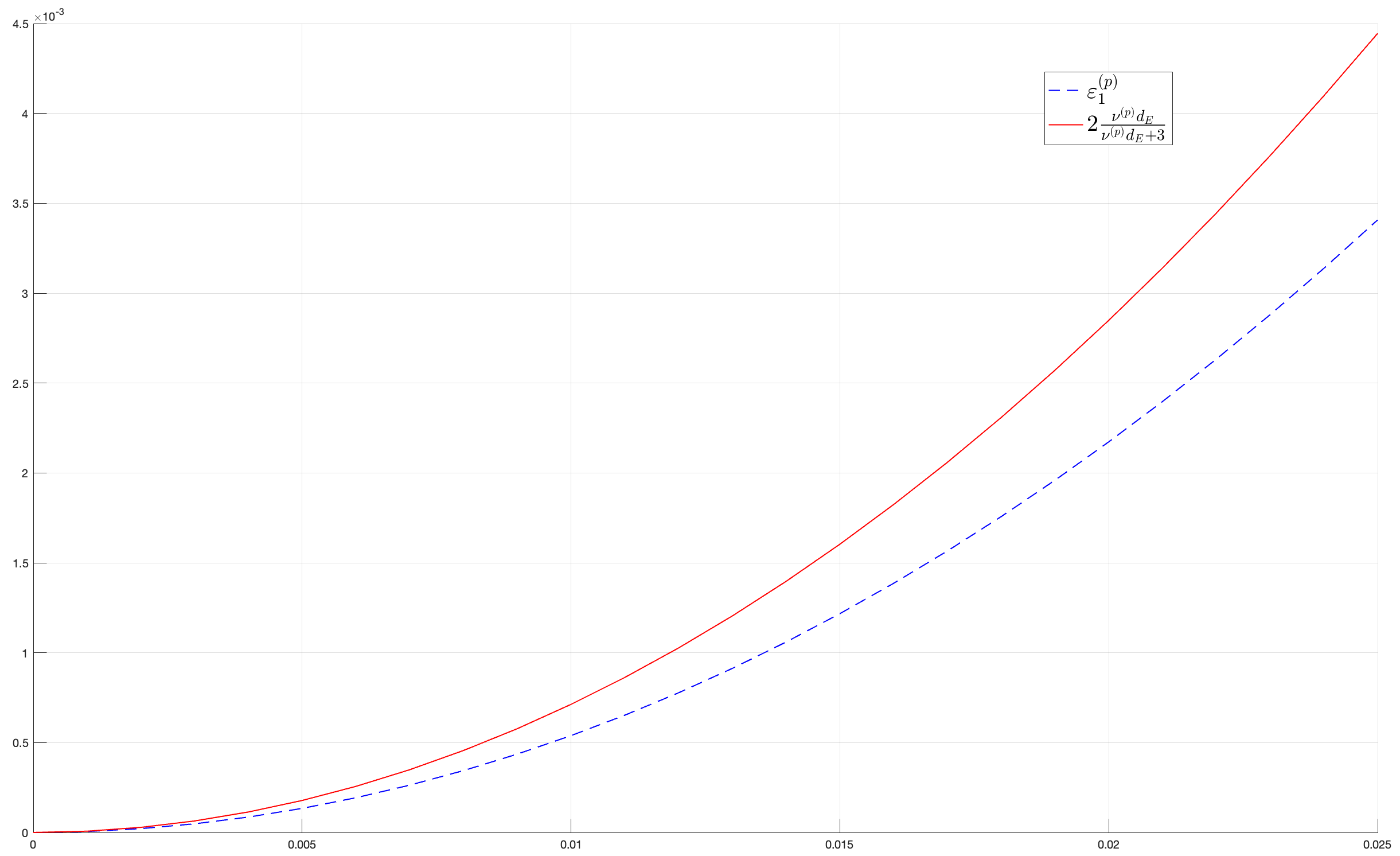}
\label{fig:norms_const}}
\caption{(a): Plots of the quantities defined {in~\eqref{eq:epdia} and~\eqref{eq:ep1nu}}. (b): Plot of $\eps^{(p)}_1$ and comparison with the quantity $2 \nu^{(p)} d_E/(\nu^{(p)} d_E + 3)$. We see that the condition $\eps^{(p)}_1 \leq 2 \nu^{(p)} d_E/(\nu^{(p)} d_E + 3)$ is indeed satisfied.\label{fig:norms_both}}
\end{figure}

\section{Plots of the right-hand side of the continuity bound for classical systems} \label{sec:appC}

{In Fig.~\ref{fig:RHSall} we present some plots of~\eqref{eq:RHSBound} for $d=10$, to illustrate the fact that, in general, there does not exist a $\eps^* \in [0,1]$ for which the function $f_{d}(\eps,\nu)$ of~\eqref{eq:RHSBound} is monotonically increasing in $\eps \in [0,\eps^*]$, for any fixed $\nu \in [0,\eps^*]$ (see Section~\ref{sec:discussions}).}

\begin{figure}[!b]
\centering
\subfloat[]{\includegraphics[width=\columnwidth]{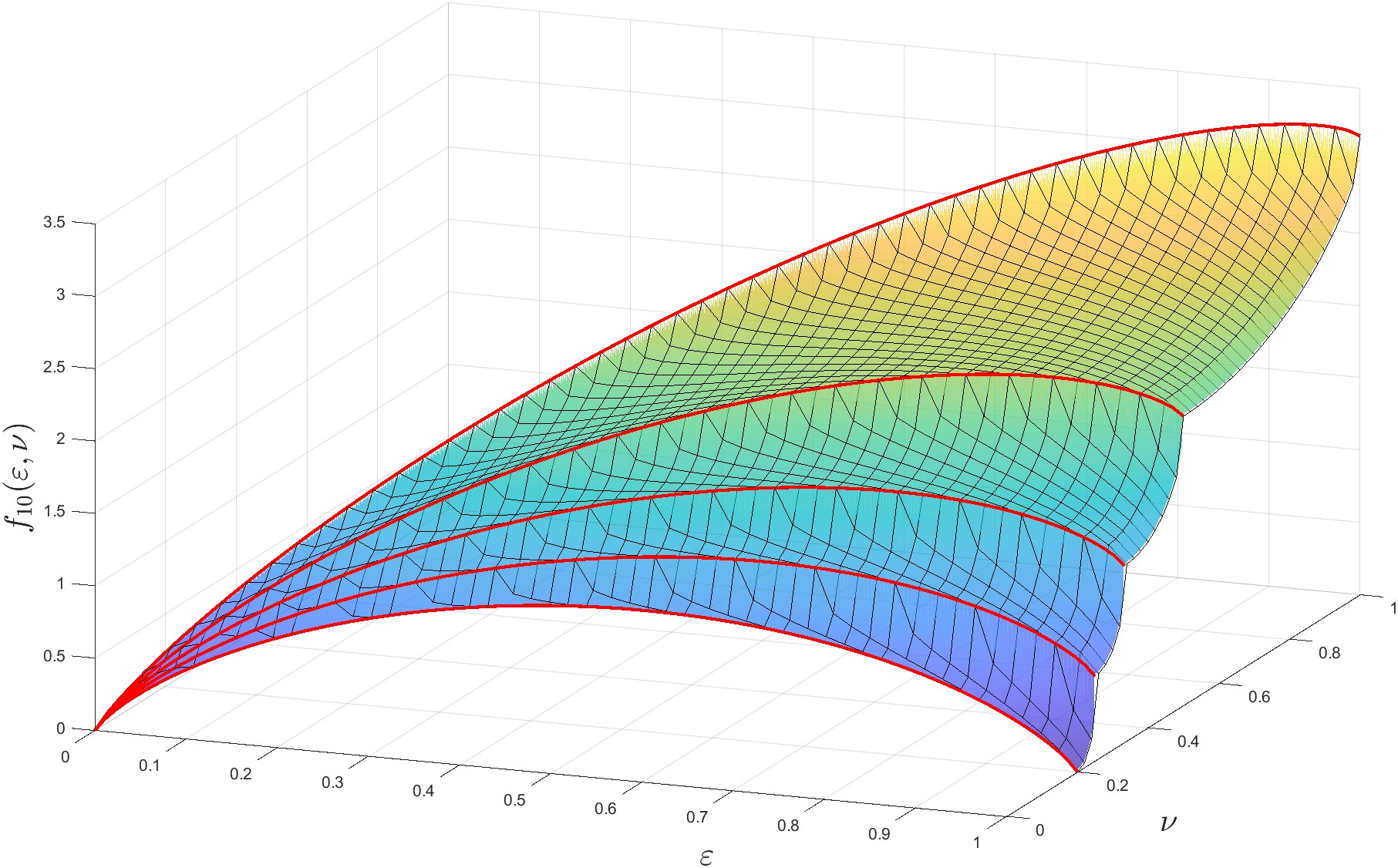}
\label{fig:RHSa}}
\hfil
\vspace{.5cm}
\subfloat[]{\includegraphics[width=\columnwidth]{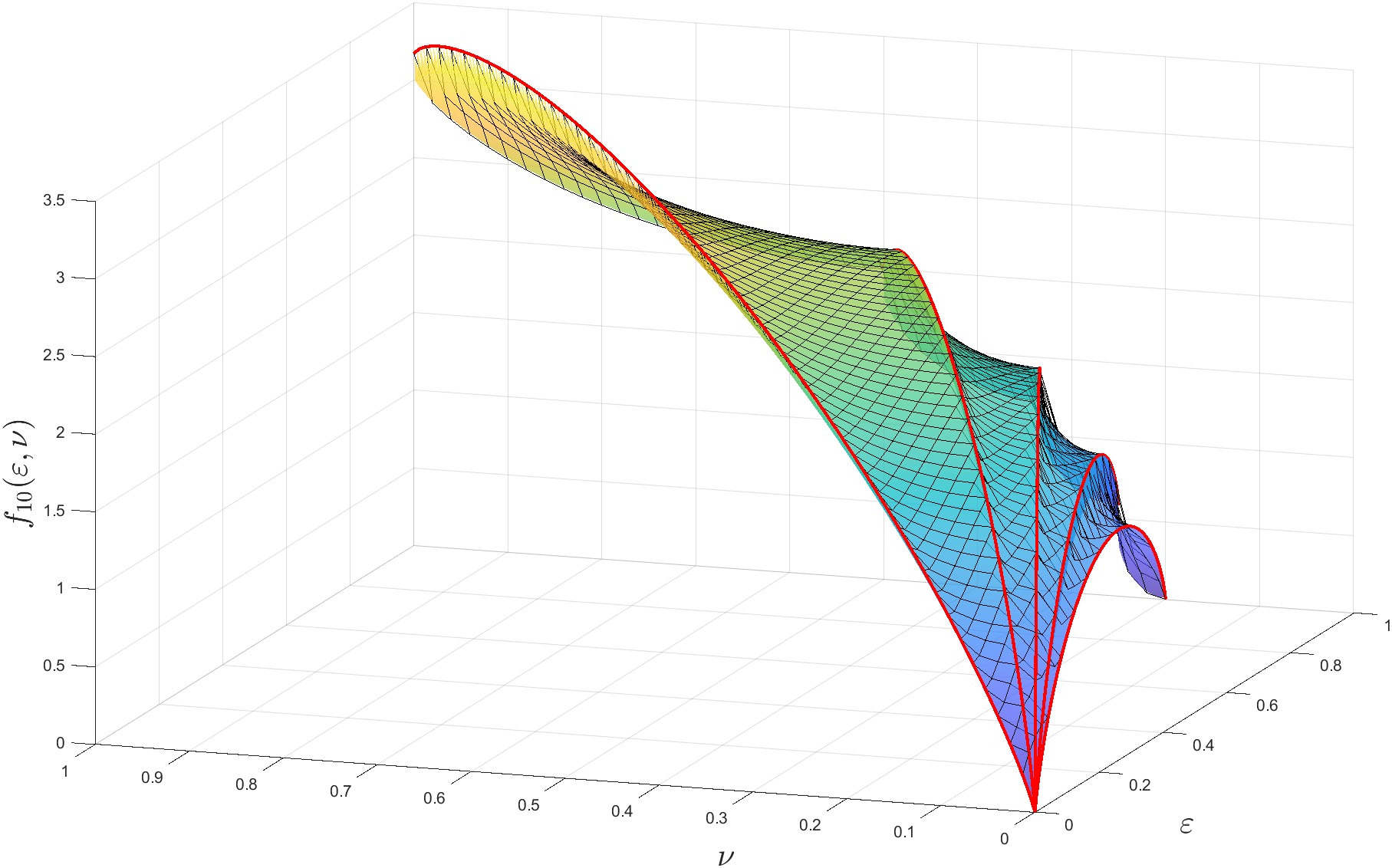}
\label{fig:RHSb}}
\hfil
\vspace{.5cm}
\subfloat[]{\includegraphics[width=\columnwidth]{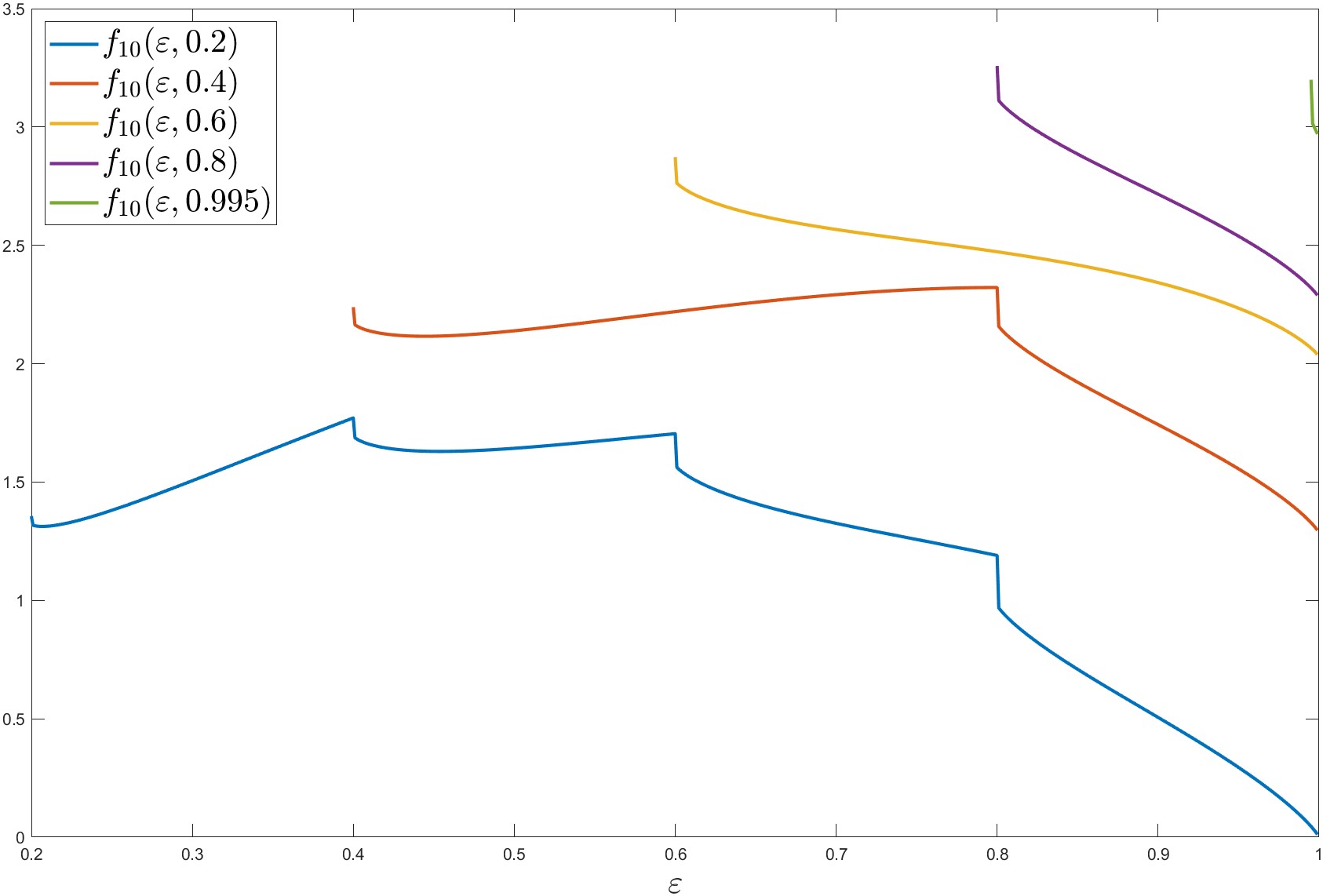}
\label{fig:RHS2}}
\caption{(a) and (b): Graphs {of~\eqref{eq:RHSBound}} for $d=10$. The two figures show the same plot from two different angles. It can be seen that if one fixes a value of $\nu$, the corresponding value of $f_{10}(\eps,\nu)$ will not necessarily be increasing with $\eps$. (c): Plots {of~\eqref{eq:RHSBound}} for $d=10$, for different fixed values of $\nu$. As it can be seen from the plots, the function $f_{d}(\eps,\nu)$ is not monotonically increasing in $\eps$ for a fixed value of $\nu$.\label{fig:RHSall}}
\end{figure}

\section*{Acknowledgements}
The authors are grateful to Andr\'as Gily\'en for making a crucial observation which proved very helpful. They also thank Koenraad Audenaert, Hamza Fawzi, Carles Roch i Carceller and John Watrous for interesting exchanges, and Mark Wilde for helpful feedback.
The hospitality of the International Centre for Mathematical Sciences (ICMS) during the workshop “Mathematical Physics in Quantum Technology: From Finite to Infinite Dimensions” (22-26 May 2023, Edinburgh) is gratefully acknowledged.
M.G.J. also acknowledges support from the Fonds de la Recherche Scientifique – FNRS.


\bibliographystyle{IEEEtran}
\bibliography{BiblioInfoTheory}

\end{document}